\documentclass[12pt,a4paper]{amsproc}

\usepackage[margin=20mm]{geometry}
\usepackage{amsmath,amssymb,amsthm}
\usepackage{mathrsfs}
\usepackage{diagbox}
\usepackage{graphicx}
\usepackage{color}

\newtheorem{theorem}{Theorem}
\newtheorem{corollary}{Corollary}
\newtheorem{lemma}{Lemma}
\newtheorem{proposition}{Proposition}
\theoremstyle{definition}
\newtheorem{assumption}{Assumption}
\newtheorem{remark}{Remark}

\numberwithin{equation}{section}

\DeclareMathOperator*{\argmin}{arg\,min}
\DeclareMathOperator*{\Var}{Var}

\def\R{\mathbb{R}}
\def\1{\mathbf{1}}
\def\ve{\varepsilon}
\def\hve{\hat{\ve}}
\def\sj{\sum_{j=1}^n}

\def\avj{\frac1n\sj}

\def\savj{n^{-1/2}\sj}
\def\opn{o_p(n^{-1/2})}
\def\op{o_p(1)}
\def\Dset{\mathbb{D}}
\def\GG{\mathscr{G}}
\def\FF{\mathscr{F}}
\def\RR{\mathscr{R}}
\def\SS{\mathscr{S}}
\def\BB{\mathscr{B}}

\def\rhat{\hat{r}}
\def\shat{\hat{\sigma}}
\def\ahat{\hat{a}}
\def\edf{\mathbb{F}}
\def\hedf{\hat{\edf}}
\def\ellvec{\mathbf{\mathrm{l}}}
\def\kvec{\mathbf{k}}
\def\evec{\mathbf{e}}
\def\Fhat{\hat{F}}
\def\Fclass{\mathfrak{F}}
\def\fmem{\mathfrak{f}}

\allowdisplaybreaks

\begin{document} 

\title[Estimating the error distribution function]{Efficient
  estimation of the error distribution function in heteroskedastic
  nonparametric regression with missing data}

\author[J.\ Chown]{Justin Chown}
\thanks{{\em Ruhr-Universit\"at Bochum, Fakult\"at f\"ur Mathematik,
    Lehrstuhl f\"ur Stochastik, 44780 Bochum, DE \\
    \indent Email: justin.chown@ruhr-uni-bochum.de
}}

\begin{abstract}
A residual-based empirical distribution function is proposed to
estimate the distribution function of the errors of a heteroskedastic
nonparametric regression with responses missing at random based on
completely observed data, and this estimator is shown to be
asymptotically most precise. 
\end{abstract}
\bigskip

\maketitle

\noindent {\em keywords:}
efficient estimator, empirical distribution function, heteroskedasticity,
local polynomial smoother, nonparametric regression, transfer principle
\bigskip



\section{Introduction}
\label{Intro}
An important problem encountered in practice occurs when variation in
the data is found to be dynamic. A typical example is when responses
$Y$ are regressed onto a vector of $m$ covariates $X$ and the errors
of that regression have variation changing in $X$. Under this
condition, many statistical procedures no longer provide consistent
inference. For example, consider a study of crop yields under
different application amounts of a fertilizer. When the variation in
yields depends on the amount of fertilizer applied, the classical
F-test will no longer provide a consistent method of inference for a
regression of these crop yields toward the amount of fertilizer
applied because it assumes the model errors have constant
variation. Examples of heteroskedastic data may be found in Greene
(2000), Vinod (2008), Sheather (2009) and Asteriou and Hall (2011).

We are interested in the case where the responses are missing. This
means observing a random sample $(X_1,\delta_1 Y_1,\delta_1)$,
$\ldots$, $(X_n,\delta_n Y_n, \delta_n)$ of data that is composed of
independent and identically distributed copies of a base observation
$(X, \delta Y, \delta)$. Here $\delta$ is an indicator variable taking
values one, when $Y$ is observed, and zero, otherwise. Throughout this
article, we will interpret a datum $(X,0,0)$ as corresponding to a
categorically missing response, i.e.\ when $\delta = 0$, the first
zero in the datum only describes the product $0\times Y = 0$, almost
surely, because we make the common assumption that $P(|Y| < \infty) =
1$. For this work, we make the following common assumption, which
ensures good performance of the nonparametric function estimators
studied in this article, concerning the covariates $X$:
\begin{assumption} \label{AssumpG}
The covariate vector $X$ has a distribution that is quasi-uniform on
the cube $[0,1]^m$; i.e.\ $X$ has a density that is both bounded and
bounded away from zero on $[0,1]^m$. 
\end{assumption}

We assume the responses are {\em missing at random} (MAR), and,
paraphrasing Chown and M\"uller (2013), we will refer to the
probability model with responses missing at random as the MAR
model. This means the distribution of $\delta$ given both the
covariates $X$ and the response $Y$ depends only on the covariates
$X$, i.e.\
\begin{equation} \label{MARassump} 
P(\delta = 1|X,Y) = P(\delta = 1|X) = \pi(X).
\end{equation}
Since we do observe some responses $Y$, we will assume that $\pi$ is
almost everywhere bounded away from zero on $[0,\,1]^m$. It is then
clear that $E\delta = E[\pi(X)]$ is positive. The MAR assumption is
commonly used and it is very reasonable in many missing data
situations (see Chapter 1 of Little and Rubin, 2002).

In this article we study the heteroskedastic nonparametric regression
model
\begin{equation} \label{modeleq}
Y = r(X) + \sigma(X)e,
\end{equation}
with the error $e$ independent of the covariate vector $X$. In order
to identify the functions $r$ and $\sigma$, we will additionally
assume the error $e$ has mean zero and unit variance. For this work,
we are interested in the case of smooth functions $r$ and $\sigma$
(see below for an explicit definition), and we will assume that
$\sigma$ is a positive--valued function so that it is a well--defined
scale function. Hence, the model above is a well--defined
heteroskedastic nonparametric regression model with identifiable
components. This model is closely related to that studied in Chown and
M\"uller (2013), who study the simpler case of $\sigma(\cdot) \equiv
\sigma_0$, a positive constant, i.e.\ $\sigma(x) = \sigma_0$ for
almost every $x$. Consequently, many results will be familiar. Here we
will need to estimate the two unknown functions $r$ and $\sigma$ with
nonparametric function estimators that are constructed from the
assumed smoothness properties of these functions. We will then use
these estimates in the proposed estimator of the distribution function
of the errors $F$.

To begin, we first consider \eqref{MARassump} and observe that 
\begin{equation*}
E[\delta h(e)] = E\delta E[h(e)]
\quad\text{and}\quad
E[\delta h(e)|X] = \pi(X)E[h(e)]
\end{equation*}
for suitable measurable functions $h$. The relations above naturally
lead to complete case estimators for each of $F$, $r$ and $\sigma$. We
investigate the residual--based empirical distribution function,
$\hedf_c$, given as
\begin{gather} \label{Fhat}
\hedf_c(t) = \frac1N\sj \delta_j \1\big[\hve_{j,c} \leq t\big] 
 = \frac1N\sj \delta_j
 \1\bigg[\frac{Y_j - \rhat_c(X_j)}{\shat_c(X_j)} \leq t \bigg], 
 \quad t \in \R.
\end{gather}
Here $N = \sj \delta_j$ is the number of completely observed pairs
$(X,Y)$ and the subscript ``$c$'' indicates the estimator is based on
the subsample of complete cases described below, which is, in general,
different from the original sample of data. Similar to the estimator
of Chown and M\"uller (2013), this is a complete case estimator.

To better explain the idea, first consider the sample $(X_1, \delta_1
Y_1, \delta_1), \ldots, (X_n, \delta_n Y_n, \delta_n)$. Due to the
i.i.d.\ nature of the this sample, we can order it in any
configuration such that triples with $\delta_j = 1$ are listed before
triples with $\delta_j = 0$. This means we can write the sample as two
subsamples:
$(X_1,Y_1,1), \ldots, (X_N,Y_N,1)$ and
$(X_{N+1},0,0),\ldots,(X_n,0,0)$. The first subsample is called the
{\em complete cases} and $N \leq n$ is the random size of the complete
cases. Hence, the estimators studied in this article use only the part
of the original sample where the responses $Y_j$ are actually
observed. For the estimator $\hedf_c$ given in \eqref{Fhat}, this
means we use only the residuals that can actually be constructed
directly from the observed responses:
\begin{equation*}
\hve_{j,c} = \frac{Y_j - \rhat_c(X_j)}{\shat_c(X_j)},
 \qquad j=1,\ldots,N,
\end{equation*}
where $\rhat_c$ and $\shat_c$ are each respectively suitable complete
case estimators of the regression function $r$ and the scale function
$\sigma$. Since we are only using a part of the original data based on
the auxiliary information that $\delta = 1$, which now has different
stochastic properties than the original data, we will, nevertheless,
argue below that $\hedf_c$ is both a consistent and an efficient
estimator for $F$.

In this work, we use local polynomial estimators of the first and
second conditional moments $r(x) = E(Y|X=x)$ and $r_2(x) =
E(Y^2|X=x)$, respectively, which we will use later to construct our
estimators $\rhat_c$ and $\shat_c$. Local polynomial estimation
follows naturally by a Taylor expansion argument, and, therefore,
follows from both of the functions $r$ and $\sigma$ satisfying certain
smoothness conditions; i.e.\ we assume both $r$ and $\sigma$ lie on
the H\"older space of functions $H(d,\varphi)$ with domain
$[0,\,1]^m$. Paraphrasing M\"uller, Schick and Wefelmeyer (2009), a
function from $[0,\,1]^m$ to $\R$ belongs to $H(d,\varphi)$, if it has
continuous partial derivatives up to order $d$ and the partial
derivatives of order $d$ are H\"older with exponent $0 < \varphi \leq
1$. We will write $H_1(d,\varphi)$ for the unit ball of $H(d,\varphi)$
(see M\"uller, Schick and Wefelmeyer, 2009, for an explicit
definition).

To define the local polynomial estimators of degree $d$, we first
introduce some notation. Let $I(d)$ be the set of multi-indices $i =
(i_1,\ldots,i_m)$ such that $i_1 + \dots + i_m \leq d$. These
multi-indices correspond with the partial derivatives of $r$ and $r_2$
(and hence $\sigma$) whose order is at most $d$. The local polynomial
estimators of $r$ and $r_2$ are respectively given by $\hat
\gamma_{a,0}$, for $a = 1,2$, where $\hat \gamma_{a,0}$ denotes the $0
= (0,\ldots,0)$ entry of the vector
\begin{equation*}
\hat \gamma_a = \argmin_{\gamma = (\gamma_i)_{i \in I(d)}} \sj\delta_j 
 \bigg\{Y_j^a - \sum_{i \in I(d)} \gamma_i
 \psi_i\bigg(\frac{X_j - x}{\lambda_n}\bigg)\bigg\}^2
 w\bigg(\frac{X_j - x}{\lambda_n}\bigg), \quad a = 1,2.
\end{equation*}
Here 
\begin{equation*}
\psi_i(x) = \frac{x_1^{i_1}}{i_1!}\cdots\frac{x_m^{i_m}}{i_m!}, 
\qquad x=(x_1,\dots,x_m) \in [0,\,1]^m, 
\end{equation*}
$w(x)= w_1(x_1)\cdots w_m(x_m)$ is a product of densities, and
$\{\lambda_n\}_{n \geq 1}$ is a bandwidth sequence, i.e.\ a sequence
of positive numbers satisfying $\lambda_n \to 0$, as $n \to
\infty$. Hence, we introduce our respective function estimators of $r$
and $\sigma$ pointwise at each $x \in [0,\,1]^m$ as $\rhat_c(x) = \hat
\gamma_{1,0}$ and $\shat_c(x) = \{\hat \gamma_{2,0} - \hat
\gamma_{1,0}^2\}^{1/2}$. Note that $\delta_1,\ldots,\delta_n$ appear
in the formula above because we require only the complete cases to
estimate both $r$ and $r_2$, and the minimization procedure above is
unaffected by the proportion $\pi$ of missing data.

Neumeyer and Van Keilegom (2010) construct an estimator related to
$\hedf_c$ for the full model using local polynomial estimators of the
first and second conditional moments, i.e.\ the simpler case where
$\delta_j = 1$, $j=1,\ldots,n$. However, these authors require the
density function of the covariates $g$ to be differentiable. For our
model, we work with the conditional density function $g_1$ of the
covariates $X$ given $\delta = 1$, and this differentiability
requirement would then be imposed on $g_1$. Using the identity for the
conditional distribution function $G_1$ of the covariates $X$ given
$\delta=1$, we have $G_1(dx) = \{\pi(x)/E\delta\} G(dx)$, and we can
see that any differentiability requirements imposed on the density
function $g_1$ must also apply to $\pi$.

To alleviate this differentiability requirement, we turn our attention
to the results of M\"uller, Schick and Wefelmeyer (2007, 2009), who
impose no such requirement. Investigating the proof of Lemma 1 of
M\"uller, Schick and Wefelmeyer (2009), reveals straightforward
modifications of those results for local polynomial function
estimation in a homoskedastic model to the heteroskedastic model
considered here. Since this approach requires Assumption
\ref{AssumpG}, using the relation between $G_1$ and $G$ above and the
bounding assumption on $\pi$, we observe that $G_1$ is quasi--uniform
whenever $G$ is quasi--uniform. We arrive at the following crucial
technical corollary to Lemma 1 of M\"uller, Schick and Wefelmeyer
(2009) for the estimators $\rhat_c$ and $\shat_c$.
\begin{corollary} \label{corMSWlem1}
Let Assumption \ref{AssumpG} hold. Suppose the regression function $r$
and the scale function $\sigma$ both belong to $H(d, \varphi)$ with
domain $[0,\,1]^m$. Further suppose the error variable has mean
equal to zero, variance equal to one and a finite moment of order $q >
4s/(2s - m)$, with $s = d + \varphi > 3m/2$. Assume the missingness
proportion $\pi$ is almost everywhere bounded away from zero on
$[0,\,1]^m$, and the densities $w_1,\ldots,w_m$ are $(m+2)$--times
continuously differentiable and have compact support $[-1,1]$. Let
$\lambda_n \sim (n\log(n))^{-1/(2s)}$. Then there is a random function
$\ahat_{1,c}$, associated to the complete case local polynomial
estimate $\rhat_c$ of $r$, such that 
\begin{equation*}
P\big(\ahat_{1,c} \in H_1(m,\alpha)\big) \to 1,
\end{equation*}
for some $\alpha > 0$,
\begin{equation*}
\int_{[0,\,1]^m} \big|\ahat_{1,c}(x)\big|^{1 + b}g_1(x)\,dx = \opn,
\end{equation*}
for $b > m/(2s - m)$, and 
\begin{equation*}
\sup_{x \in [0,\,1]^m} \big|\rhat_c(x) - r(x) - \ahat_{1,c}(x)\big| 
 = \opn.
\end{equation*}
If, additionally, the error variable has a finite moment of order
$2q$, then there is a random function $\hat a_{2,c}$, associated to
the complete case local polynomial estimate $\rhat_{2,c}$ of $r_2$,
such that
\begin{equation*}
P\big(\ahat_{2,c} \in H_1(m,\alpha)\big) \to 1,
\end{equation*}
\begin{equation*}
\int_{[0,\,1]^m} \big|\ahat_{2,c}(x)\big|^{1 + b}g_1(x)\,dx = \opn
\end{equation*}
and
\begin{equation*}
\sup_{x \in [0,\,1]^m} \big|\rhat_{2,c}(x) - r_2(x) - \ahat_{2,c}(x)\big| 
 = \opn.
\end{equation*}
\end{corollary}

Paraphrasing Remark 5 of M\"uller, Schick and Wefelmeyer (2009), there
is a trade-off between the required smoothness of the regression and
scale functions (indicated by the variable $s$) and the existence of
higher order moments for the error variable $e$ (indicated by
$q$). This means that higher degree polynomials, used to approximate
$r$ and $\sigma$, require higher order moments of $e$ to
exist. Further, we can see that a larger bandwidth may be used to
estimate these functions when they are smooth but a smaller bandwidth
will be required when these functions are rough. In light of the above
results, we are able to obtain analogous conclusions to those of Lemma
A.2 of Neumeyer and Van Keilegom (2010). 
\begin{proposition} \label{proprhatshat}
Suppose the first set of assumptions of Corollary \ref{corMSWlem1} are
satisfied. Then we have
\begin{equation*}
\int_{[0,\,1]^m} \frac{\rhat_c(x) - r(x)}{\sigma(x)}g_1(x)\,dx 
 = \frac1N\sj \delta_je_j + \opn.
\end{equation*}
Now suppose the additional assumptions of Corollary \ref{corMSWlem1}
are satisfied. Then we have
\begin{equation*}
\int_{[0,\,1]^m} \frac{\hat \sigma_c(x) - \sigma(x)}{\sigma(x)}g_1(x)\,dx
 = \frac1N\sj \delta_j\frac{e_j^2 - 1}{2} + \opn.
\end{equation*}
\end{proposition}

\begin{remark} \label{remrhatrhat2}
Analogous results to Corollary \ref{corMSWlem1} above hold for the
full model where $\delta_j = 1$, $j=1,\ldots,n$, and $N = n$, almost
surely, in both cases where the covariate distribution is $G$ and
$G_1$. Here the local polynomial estimators $\rhat$ for $r$ (with
associated $\ahat_1$) and $\rhat_2$ for $r_2$ (with associated
$\ahat_2$) are respectively defined exactly as $\rhat_c$ and
$\rhat_{2,c}$ are defined above, but now the indicators
$\delta_1,\ldots,\delta_n$ are all equal to one. Hence, we obtain
estimators $\rhat$ for $r$  and $\shat$ for $\sigma$ for which
analogous results of Proposition \ref{proprhatshat} hold in the full
model in both cases where the covariate distribution is $G$ and
$G_1$. The case of covariates having distribution $G$ confirms the
findings of Lemma A.2 of Neumeyer and Van Keilegom (2010), which are
required to prove their main result.
\end{remark}

As noted in Remark 1 of Chown and M\"uller (2013), one proves the
above statements analogously to how M\"uller, Schick and Wefelmeyer
(2009) prove their results (inspect the proof of Lemma 1 of that
paper). The only changes occur by introducing $\sigma$ and the
indicators $\delta_1,\ldots,\delta_n$. Since we also estimate $r_2$,
this requires strengthening the moment conditions on the error
variable $e$ from $q$ to $2q$ because $Y^2 = r^2(X) + \sigma^2(X) +
2r(X)\sigma(X)e + \sigma^2(X)(e^2 - 1)$, which follows from the model
equation above. An additional requirement needed by Neumeyer and Van
Keilegom (2010) for their results to hold is that $\sup_{t \in \R}
|t^2F''(t)| < \infty$. This assumption implies the curvature of the
function space underlying the probability model is finite. However, we
can measure this curvature using Fisher information. This means we can
merely assume that $F$ has finite Fisher information for both location
and scale, written as Assumption \ref{AssumpF} below, which is a
lighter assumption than $\sup_{t \in \R} |t^2F''(t)| < \infty$. We now
arrive at our third auxiliary result, which confirms the results of
Neumeyer and Van Keilegom (2010). The proof of this result is rather
elaborate and technical. Therefore, it is held to Section
\ref{Appendix}.

\begin{assumption} \label{AssumpF}
The error density $f$ is absolutely continuous with almost everywhere 
derivative $f'$ and finite Fisher information for both location and 
scale; i.e.\
\begin{equation*}
\int_{-\infty}^{\infty} \big\{1 + z^2\big\}
 \bigg\{\frac{f'(z)}{f(z)}\bigg\}^2\,F(dz) 
 < \infty.
\end{equation*}
\end{assumption}
\begin{theorem}[{\sc expansion for the full model}] 
\label{ThmFhatFullmodel2}
Assume the covariates $X$ are distributed according to $G$. Let the
required assumptions of Proposition \ref{proprhatshat} be satisfied
concerning the local polynomial estimators $\rhat$ and $\rhat_2$ (see
Remark \ref{remrhatrhat2} above). Further, let Assumption
\ref{AssumpF} hold with the error density $f$ additionally satisfying
$\sup_{t \in \R} f(t) < \infty$ and $\sup_{t \in \R} |tf(t)| <
\infty$. Then, for  $\hve_1 = \{Y_1 - \rhat(X_1)\}/\shat(X_1),\ldots,
\hve_n = \{Y_n - \rhat(X_n)\}/\shat(X_n)$, we have
\begin{equation*}
\sup_{t \in \R} \bigg| \avj \bigg[
 \1\big[\hve_j \leq t\big] - \1\big[e_j \leq t\big]
 - f(t)\bigg\{e_j + \frac{t}{2}\big(e_j^2 - 1\big)\bigg\}\bigg] 
 \bigg| = \opn. 
\end{equation*}
\end{theorem}

We now adapt the results of Theorem \ref{ThmFhatFullmodel2} to the MAR
model using the {\em transfer principle for complete case statistics}
given in Koul, M\"uller and Schick (2012). Expanding on the
observations of Chown and M\"uller (2013), it follows that we can
factor the joint distribution of $(X,Y)$ into two components: the
distribution $G$ of the covariates $X$ and the conditional
distribution of the responses $Y$ given $X$, i.e.\ the distribution
$F$ of the errors $e$. Now, using the MAR assumption, we observe that
$Y$ and $\delta$ are independent given $X$. This implies only the
distribution $G$ changes to $G_1$ when moving from full model to the
MAR model, e.g.\ complete case statistics are based on observations
$(X,Y)$ with a joint conditional distribution given $\delta = 1$,
which can now be factored into $G_1$ and $F$. Hence, the functionals
$F$, $r$ and $\sigma$ remain the same in the MAR model. This implies
the complete case statistic $\hedf_c$ is a consistent estimator for
$F$ in the MAR model. However, in order to apply the transfer
principle, we need to restate the result of Theorem
\ref{ThmFhatFullmodel2} for covariates that have distribution $G_1$,
which corresponds to the data used in our complete case estimator
$\hedf_c$. The proof of this result follows immediately from the proof
of Theorem \ref{ThmFhatFullmodel2} (see Section \ref{Appendix})
with the discussion in Remark \ref{remrhatrhat2}
above.
\begin{corollary}[{\sc expansion for the full model using $G_1$}]
\label{corFhatExpanG1} Assume the covariates $X$ are distributed
according to $G_1$, and $\pi$ is almost everywhere bounded away from
zero on $[0,\,1]^m$. Let the required assumptions of Proposition
\ref{proprhatshat} be satisfied concerning the local polynomial
estimators $\rhat$ and $\rhat_2$ (see Remark \ref{remrhatrhat2}
above). Further, let Assumption \ref{AssumpF} hold with the error
density $f$ additionally satisfying $\sup_{t \in \R} f(t) < \infty$
and $\sup_{t \in \R} |tf(t)| < \infty$. Then, for  $\hve_1
= \{Y_1 - \rhat(X_1)\}/\shat(X_1),\ldots,\hve_n = \{Y_n -
\rhat(X_n)\}/\shat(X_n)$, we have
\begin{equation*}
\sup_{t \in \R} \bigg| \avj \bigg[
 \1\big[\hve_j \leq t\big] - \1\big[e_j \leq t\big]
 - f(t)\bigg\{e_j + \frac{t}{2}\big(e_j^2 - 1\big)\bigg\}\bigg] 
 \bigg| = \opn. 
\end{equation*}
\end{corollary}

Combining the results above with the transfer principle for complete
case statistics, we can immediately derive the expansion of our
complete case estimator $\hedf_c$. We investigate the efficiency bound
for regular estimators of $F$ in the MAR model in Section
\ref{Efficiency} below, i.e.\ estimators whose limit distributions do
not depend on any direction of approach. In Corollary \ref{corFeff}
(see Section \ref{Efficiency} below), we provide the efficient
influence function characterizing the class of efficient estimators of
$F$ in the MAR model. Since the influence function of our complete
case estimator $\hedf_c$ matches the efficient influence function,
this characterizes $\hedf_c$ as an efficient estimator for $F$ in the
MAR model, which implies that $\hedf_c$ is an asymptotically most
precise (least dispersed) estimator. We now arrive at the main result
of this section:
\begin{theorem}[{\sc expansion for the MAR model}] \label{ThmFhatMARmodel}
Consider the heteroskedastic nonparametric regression model with responses 
missing at random. Let Assumption \ref{AssumpF} hold with the error
density $f$ additionally satisfying $\sup_{t \in \R} f(t) < \infty$
and $\sup_{t \in \R} |tf(t)| < \infty$, and let the assumptions of
Corollary \ref{corMSWlem1} hold. Then the complete case estimator
$\hedf_c$ of the error distribution function $F$ satisfies the uniform
stochastic expansion
\begin{equation*} 
\sup_{t \in \R} \bigg| \hedf_c 
 - \frac1N\sj \delta_j \1\big[e_j \leq t\big]
 - \frac1N\sj \delta_j f(t)
 \bigg\{e_j + \frac{t}{2}\big(e_j^2 - 1\big)\bigg\}
 \bigg| = \opn.
\end{equation*}
Furthermore, $\hedf_c$ is asymptotically linear, uniformly in $t \in
\R$, with influence function 
\begin{equation*}
\phi(\delta, e, t) = \frac{\delta}{E\delta}\bigg[
 \1\big[e \leq t\big] - F(t)
 + f(t)\bigg\{e + \frac{t}{2}\big(e^2 - 1\big)\bigg\}\bigg], 
\end{equation*}
and $\hedf_c$ is asymptotically efficient, in the sense of H\'ajek and
Le Cam, for estimating $F$. 
\end{theorem}
\begin{proof}
The assumptions of Theorem \ref{ThmFhatFullmodel2} and Corollary
\ref{corFhatExpanG1} are satisfied. Hence, for the full model, we have
\begin{equation*}
\sup_{t \in \R} \bigg| \avj \bigg[
 \1\big[\hve_j \leq t\big] - \1\big[e_j \leq t\big]
 - f(t)\bigg\{e_j + \frac{t}{2}\big(e_j^2 - 1\big)\bigg\}\bigg] \bigg| 
 = \opn,
\end{equation*}
when the covariates $X$ are distributed under either $G$ or
$G_1$. Since $\hedf_c$ is the complete case version of the estimator
in the display above, it follows from Remark 2.5 of Koul, M\"uller
and Schick (2012) for the first assertion to hold, i.e.\
\begin{equation*} 
\sup_{t \in \R} \bigg| \hedf_c 
 - \frac1N\sj \delta_j \1\big[e_j \leq t\big]
 - \frac1N\sj \delta_j f(t)
 \bigg\{e_j + \frac{t}{2}\big(e_j^2 - 1\big)\bigg\}
 \bigg| = \opn.
\end{equation*}
This expansion is equivalent to 
\begin{equation*}
\sup_{t \in \R} \bigg| \avj \frac{\delta_j}{E\delta}\bigg[
 \1\big[\hve_{j,c} \leq t\big] - \1\big[e_j \leq t\big]
 - f(t)\bigg\{e_j + \frac{t}{2}\big(e_j^2 - 1\big)\bigg\}\bigg] \bigg| 
 = \opn,
\end{equation*}
and we find, uniformly in $t \in \R$,
\begin{equation*}
\hedf_c(t) = \avj \frac{\delta_j}{E\delta}
 \1\big[\hve_{j,c} \leq t\big] + \opn
 = F(t) + \avj \phi(\delta_j, e_j, t) + \opn, 
\end{equation*}
where the function $\phi(\delta, e, t) = (\delta/E\delta)[\1[e \leq t] 
- F(t) + f(t)\{e + t/2(e^2 - 1)\}]$ is the influence function for 
$\hedf_c$. Since the assumptions of Corollary \ref{corFeff} in Section
\ref{Efficiency} below are satisfied, it follows for the influence
function $\phi$ to be the efficient influence function for estimating
$F$, which concludes the proof.
\end{proof}

We note the uniform expansion above implies the existence of a
functional central limit theorem. In addition, the property that
$\hedf_c$ is efficient means that competing estimators will not
achieve higher precision for large samples. This includes estimators
that employ imputation approaches to estimate the missing responses. A
consequence of this conclusion is that imputation procedures employed
to estimate $F$ may only be effective in small samples. Therefore, we
recommend the use of the complete case estimator $\hedf_c$ for
conducting various hypothesis tests concerning the heteroskedastic MAR
model. Section \ref{Efficiency} details the remaining results
necessary for proving Theorem \ref{ThmFhatMARmodel}. Section
\ref{Application} concludes the article with a numerical study of the
previous results.


\section{Efficiency}
\label{Efficiency}
In this section we will construct the efficient influence function for
estimating a linear functional $E[h(e)]$ based on observations of the
form $(X,\delta Y, \delta)$, and later specialize this result to $F(t)
= E[\1[e \leq t]]$, $t \in \R$. We will first follow the arguments of
Chown and M\"uller (2013), who study this problem for the special case
of a constant variance function. In addition, we follow the arguments
of M\"uller, Schick and Wefelmeyer (2006), who consider linear
functionals of the joint distribution of $X$ and $Y$ with data of the
above form. Finally, we use insight from the arguments of Schick
(1994), who study estimation of functionals from various
heteroskedastic regression models. We only summarize their main
arguments and refer the reader to these papers for further
details. This allows us to adapt parts of those proofs to the model
considered here. Consequently, we only sketch the proofs of the
results in this section. To continue, we require Assumption
\ref{AssumpF} to hold.

In the following, no assumption of a parametric model (finite
dimensional) is imposed on any of the regression function, the scale
function or the joint distribution of the observations. This means the
parameter set $\Theta$ consists of the unknown functions of the
statistical model: a family of covariate distributions $\GG$
satisfying Assumption \ref{AssumpG}, a family of error distributions
$\FF$ that have mean zero, unit variance, finite fourth moment and
satisfy Assumption \ref{AssumpF}, a space of regression functions
$\RR$ that belong to $H(d,\varphi)$, a space of scale functions $\SS$
that is a subspace of $\RR$ composed of positive--valued functions and
a family of response probability distributions $\BB$ that are
characterized by the functions from $[0,\,1]^m$ to $(0,1]$. More
precisely, $\Theta = \GG \times \FF \times \RR \times \SS \times
\BB$.

We now proceed as in Section 2 of Chown and M\"uller (2013). Since the
construction of the efficient influence function utilizes directional
information in $\Theta$, we now identify the set of perturbations
$\dot \Theta$, which may be thought of as directions. Observe the
joint distribution $P(dx,dy,dz)$ takes the form
\begin{equation*}
P(dx,dy,dz) = G(dx)B_{\pi(x)}(dz)\Big\{zQ(dy|x) + (1-z)\delta_0(dy)\Big\},
\end{equation*}
where $B_p = p\delta_1 + (1-p)\delta_0$ denotes the Bernoulli
distribution with parameter $p$ and $\delta_t$ as the Dirac measure at
$t$. The model considered here deviates from that considered in Chown
and M\"uller (2013) only in the conditional distribution $Q$ of $Y$
given $X$. This means we first need to consider the spaces
$\mathcal{L}_{2,0}(G)$, $\mathcal{L}_2(G_{\pi})$ and
$\mathcal{V}_0$. Here $\mathcal{L}_{2,0}(G)$ is the space of functions
that are square integrable and have mean zero with respect to $G$,
$\mathcal{L}_2(G_{\pi})$ is a subspace of $\mathcal{L}_2(G)$, where the
functions $w$ now satisfy $E[w^2(X)\pi(X)\{1 - \pi(X)\}] < \infty$, and
$\mathcal{V}_0$ is the space of functions satisfying $\int
v(x,y)Q(dy|x) = 0$. It then follows for perturbations $G_{nu}$ of $G$,
$\pi_{nw}$ of $\pi$ and $Q_{nv}$ of $Q$ that are Hellinger
differentiable requires the functions $u$, $w$ and $v$ to be further
restricted to appropriate subspaces. Since we have only assumed a
model for $Q$, which follows from the heteroskedastic nonparametric
regression model \eqref{modeleq}, this only requires resolving the
appropriate subspace $\mathcal{V}$ of $\mathcal{V}_0$.

Using the independence of the covariates $X$ and errors $e$, we may
write
\begin{equation*}
\frac{d}{dy}Q(y|x) = 
 f\bigg(\frac{y - r(x)}{\sigma(x)}\bigg)\frac{1}{\sigma(x)}.
\end{equation*}
Hence, in order to derive the explicit form of $\mathcal{V}$, we
introduce further perturbations $s$, $t$ and $m$ of the unknown
functions $f$, $r$ and $\sigma$, respectively, and write
\begin{equation*}
\frac{d}{dy}Q_{nv}(y|x) = \frac{d}{dy}Q_{nstm}(y|x)
 = f_{ns}\bigg(\frac{y - r_{nt}(x)}{\sigma_{nm}(x)}\bigg)
 \frac1{\sigma_{nm}(x)},
\end{equation*}
where $f_{ns}(z) = f(z)\{1 + n^{-1/2}s(z)\}$, $r_{nt}(x) = r(x) +
n^{-1/2}t(x)$ and $\sigma_{nm}(x) = \sigma(x) + n^{-1/2}m(x)$ for $s
\in \mathcal{S}$, $t \in \mathcal{L}_2(G_1)$ and $m \in
\mathcal{L}_2(G_1)$. Here
\begin{equation*} 
\mathcal{S} = \bigg\{ s \in \mathcal{L}_2(F) \,:\,
 \int_{-\infty}^{\infty} s(z)f(z)\,dz = 0,
 \int_{-\infty}^{\infty} zs(z)f(z)\,dz = 0
 \text{ and }
 \int_{-\infty}^{\infty} z^2s(z)f(z)\,dz = 0 \bigg\},
\end{equation*}
which is derived by the constraints that $f_{ns}$ must integrate to
one, have mean zero and have unit variance. In the following we will
write ``$\doteq$'' to denote asymptotic equivalence; i.e.\ equality up
to an additive term of order $\opn$. In addition, we introduce the
notation $\ellvec(z) = (\ell_1(z),\ell_2(z))^T$, for $\ell_1(z) =
-f'(z)/f(z)$ and $\ell_2(z) = -1 - zf'(z)/f(z)$, $\kvec(x) =
(t(x)/\sigma(x),m(x)/\sigma(x))^T$, $\evec_1 = (1,0)^T$ and $\evec_2 =
(0,1)^T$. Similar to the calculations of Chown and M\"uller (2013) and
Schick (1994), who considers, more generally, directionally
differentiable regression and scale functions, we have, by a brief
sketch,
\begin{align*} 
&f_{ns}\bigg(\frac{y - r_{nt}(x)}{\sigma_{nm}(x)}\bigg)
 \frac{1}{\sigma_{nm}(x)} \\
&\quad\doteq f\bigg(\frac{y - r(x)}{\sigma(x)}\bigg)
 \frac{1}{\sigma(x)} \times \bigg\{1 + n^{-1/2}\bigg[
 \kvec^T(x)\ellvec\bigg(\frac{y - r(x)}{\sigma(x)}\bigg)
 + s\bigg(\frac{y - r(x)}{\sigma(x)}\bigg)\bigg]\bigg\}.
\end{align*} 
Hence, 
\begin{equation*}
\frac{d}{dy}Q_{ns\kvec}(y|x) \doteq
 f\bigg(\frac{y - r(x)}{\sigma(x)}\bigg)\frac{1}{\sigma(x)}
 \bigg\{1 + n^{-1/2}\bigg[
 \kvec^T(x)\ellvec\bigg(\frac{y - r(x)}{\sigma(x)}\bigg)
 + s\bigg(\frac{y - r(x)}{\sigma(x)}\bigg)\bigg]\bigg\} 
\end{equation*}
and $\mathcal{V}$ takes the form
\begin{equation*}
\mathcal{V} = \bigg\{
v(x,y) = \kvec^T(x)\ellvec\bigg(\frac{y - r(x)}{\sigma(x)}\bigg)
 + s\bigg(\frac{y - r(x)}{\sigma(x)}\bigg):
 \kvec \in \mathcal{L}_2(G_1)\times\mathcal{L}_2(G_1) \text{ and } 
 s \in \mathcal{S}\bigg\}.
\end{equation*}

Thus we have perturbations $\dot \Theta = \mathcal{L}_{2,0}(G) \times
\mathcal{S} \times \{\mathcal{L}_2(G_1) \times \mathcal{L}_2(G_1)\}
\times \mathcal{L}_2(G_\pi)$. Observe, for any $\gamma =
(u,s,\kvec,w)$ in $\dot \Theta$, the perturbed distribution
$P_{n\gamma}(dx,dy,dz)$ of an observation $(X,\delta Y,\delta)$ is
then
\begin{equation*} 
P_{n\gamma}(dx,dy,dz) = G_{nu}(dx)B_{\pi_{nw}(x)}(dz)
 \Big\{zQ_{ns\kvec}(dy|x) + (1 - z)\delta_0(dy)\Big\}.
\end{equation*}
It follows that $P$ is Hellinger differentiable with tangent
\begin{equation*} 
d_\gamma\big(X,\delta Y, \delta\big) = u(X) 
 + \big\{\delta - \pi(X)\big\}w(X)
 + \delta\big\{\kvec^T(X)\ellvec(e) + s(e)\big\}, 
\end{equation*}
and we arrive at the form of the tangent space as
\begin{equation*}
T = \mathcal{L}_{2,0}(G)
 \oplus \Big\{\big\{\delta - \pi(X)\big\}w(X) 
 \,:\, w \in \mathcal{L}_2(G_\pi)\Big\}
 \oplus \Big\{\delta v(X,Y): v \in \mathcal{V}\Big\}.
\end{equation*}
Consequently, we have local asymptotic normality. This means the
following expansion holds:
\begin{equation*}
\sj \log\bigg(\frac{dP_{n\gamma}}{dP}
 \big(X_j,\delta_j Y_j, \delta_j\big)\bigg)
 = \savj d_\gamma\big(X_j,\delta_j Y_j,\delta_j\big) 
 - \frac12 E\Big[d_\gamma^2\big(X,\delta Y,\delta\big)\Big] + \op,
\end{equation*}
where $dP$ denotes the density function of $P$.

We are interested in the linear functional $E[h(e)]$. In order to
specify a gradient for $E[h(e)]$, we first need to find its
directional derivative $\gamma_h \in \dot \Theta$, which is
characterized by a limit as follows. As in M\"uller, Schick and
Wefelmeyer (2004), we have, for every $s \in S$,
\begin{equation*}
\lim_{n \to \infty} n^{1/2}\bigg[
 \int_{-\infty}^{\infty} h(z)f_{ns}(z)\,dz
 - \int_{-\infty}^{\infty} h(z)f(z)\,dz \bigg]
 = E[h(e)s(e)] = E\big[h_0(e)s(e)\big],
\end{equation*}
with $h_0$ given as a projection of $h$ onto $\mathcal{S}$:
\begin{align*} 
h_0(z) &= h(z) - E[h(e)] - zE[e h(e)] \\
&\quad - \frac{z^2 - E[e^3]z - 1}{E[e^4] - E^2[e^3] - 1}
 \Big\{E\big[e^2 h(e)\big] - E\big[e^3\big]E[e h(e)] 
 - E[h(e)]\Big\}.
\end{align*}
Thus, $E[h(e)]$ is directionally differentiable with directional
derivative $(0,h_0,\mathbf{0},0)$ and gradient $h_0(e)$. By the
convolution theorem (see, for example, Section 2 of Schick, 1993) the
unique canonical gradient $g^*(X,\delta Y, \delta)$ is found by
orthogonally projecting the gradient $h_0(e)$ onto the tangent space
$T$. Thus, $g^*(X,\delta Y, \delta)$ must take the form
\begin{align} \label{cangrad}
g^*\big(X,\delta Y, \delta\big) = u^*(X) 
 + \big\{\delta - \pi(X)\big\}w^*(X)
 + \delta\big\{\kvec^{*T}(X)\ellvec(e) + s^*(e)\big\}.
\end{align}
Now proceeding as in Section 2 of Chown and M\"uller (2013), we obtain
the following result:
\begin{lemma} \label{lemcangrad}
The canonical gradient of $E[h(e)]$ is $g^*(X,\delta Y,\delta)$, which
is characterized by $(0,s^*,\kvec^*,0)$, where
\begin{equation*}
s^*(z) = \frac{1}{E\delta}h_0(z) - E_1\big[\kvec^{*T}(X)\big]\ellvec_0(z)
\quad \text{and} \quad 
\kvec^* \equiv -\frac{1}{E\delta}J_d^{-1}E\big[\ellvec_0(e)h_0(e)\big], 
\end{equation*}
with $h_0$ given above and the quantities
\begin{equation*}
\ellvec_0(z) = \ellvec(z) - z\evec_1
 - \frac{z^2 - E[e^3]z - 1}{E[e^4] - E^2[e^3] - 1}
 \Big\{2\evec_2 - E\big[e^3\big]\evec_1\Big\}
\end{equation*}
and
\begin{equation*}
J_d^{-1} = \frac{1}{E[e^4] - E^2[e^3] - 1}
 \begin{bmatrix}
 E[e^4] - 1 & -2E[e^3] \\
 -2E[e^3] & 4
 \end{bmatrix}.
\end{equation*}
\end{lemma}

We will call an estimator $\hat \mu$ for $E[h(e)]$ efficient, in the
sense of H\'ajek and Le Cam, if it is asymptotically linear with
corresponding influence function equal to the canonical gradient
$g^*(X,\delta Y,\delta)$ that characterizes $E[h(e)]$. This means
$\hat \mu$ satisfies the expansion
\begin{equation*}
n^{1/2}\big\{\hat{\mu} - E[h(e)]\big\} 
 = \savj g^*\big(X_j,\delta_j Y_j,\delta_j\big) + o_p(1). 
\end{equation*}
We combine this fact with Lemma \ref{lemcangrad} and \eqref{cangrad}
to obtain the following result:
\begin{theorem} \label{thmheff}
Consider the heteroskedastic nonparametric regression model with
responses missing at random. An estimator $\hat \mu$ of $E[h(e)]$ is
efficient, if it satisfies the expansion
\begin{equation*}
n^{1/2}\big\{\hat{\mu} - E[h(e)]\big\} = \savj \frac{\delta}{E\delta}
 \Big[h_0(e_j) - E^T\big[h_0(e_j)\ellvec_0(e_j)\big]J_d^{-1}
 \ellvec_d(e_j)\Big] + o_p(1), 
\end{equation*}
where $h_0$ is given above, $\ellvec_0$ and $J_d^{-1}$ are given in
Lemma \ref{lemcangrad} and
\begin{equation*}
\ellvec_d(z) = z\evec_1 + \frac{z^2 - zE[e^3] - 1}{E[e^4] - E^2[e^3] - 1}
 \Big\{2\evec_2 - E\big[e^3\big]\evec_1\Big\}. 
\end{equation*}
\end{theorem}

In this article, we are interested in the function $h(z) = \1[z \leq
t]$ because we estimate $F(t) = E[\1[e \leq t]]$ using
$\hedf_c$. We now obtain, using Theorem \ref{thmheff} with this $h$,
the expansion for an efficient estimator of the error distribution
function $F$.
\begin{corollary} \label{corFeff}
Consider the heteroskedastic nonparametric regression model with
responses missing at random. An estimator $\Fhat$ of $F$ is efficient,
in the sense of H\'ajek and Le Cam, if it satisfies the expansion
\begin{equation*}
n^{1/2}\big\{\Fhat(t) - F(t)\big\} 
 = \savj \frac{\delta}{E\delta} \bigg[\1\big[e_j \leq t\big] - F(t)
 + f(t)\bigg\{e_j + \frac{t}{2}\big(e_j^2 - 1\big)\bigg\}\bigg] 
 + o_p(1). 
\end{equation*}
\end{corollary}


\section{Simulations} \label{Application}
We conclude this article with a small numerical study of the previous
results. In the following we work with
\begin{equation*}
r\big(x_1,x_2\big) = 1 + x_1 - x_2 + 2e^{-\frac12\sqrt{x_1^2 + x_2^2}}
\quad\text{and}\quad
\sigma\big(x_1,x_2\big) = \sqrt{1 + 2x_1^2 + 2x_2^2}
\end{equation*}
to preserve the nonparametric nature of the study. The covariates
$X_1$ and $X_2$ are each randomly generated from a $U(-1,1)$
distribution, and the errors $e$ are generated from a standard normal
distribution. The indicators $\delta$ are randomly generated from a
Bernoulli$(\pi(X_1,X_2))$ distribution, with $\pi(X_1,X_2) =
P(\delta=1|X_1,X_2)$. Here we use $\pi(x_1,x_2) = 1 - 1/(1 +
e^{-(x_1 + x_2)/2})$. Consequently, the average amount of missing
data is about 50\% (ranging between 26\% and 74\%). We work with
$d=3$, the locally cubic smoother, to estimate both of the functions
$r$ and $\sigma$. For our choice of using a product of tricubic kernel
functions and bandwidth $\lambda_n = 3(n\log(n))^{-1/7}$, the
assumptions of Theorem \ref{ThmFhatMARmodel} are satisfied.

\begin{table}
\begin{minipage}{\textwidth}
\centering
\begin{tabular}{|r|c c c c|}
\hline
\diagbox{$n$}{$t$} & $-3$ & $-2$ & $-1$ & $0$ \\
\hline
100 & 0.0977 (0.0231) & 0.0563 (0.0588) & 0.0030 (0.1516) & -0.0318
 (0.1957) \\
200 & 0.0965 (0.0241) & 0.0777 (0.0722) & 0.1646 (0.1555) & -0.0818
 (0.2124) \\
500 & 0.0301 (0.0089) & 0.0008 (0.0496) & 0.1806 (0.1285) & -0.0746
 (0.2022) \\
1000 & 0.0006 (0.0030) & -0.0382 (0.0348) & 0.1389 (0.1033) & -0.0826
 (0.1848) \\
\hline
\end{tabular}
\vspace{1ex}
\caption{{\em Simulated asymptotic bias and variance (in parentheses),
    at the points $-3$, $-2$, $-1$ and $0$, of $n^{-1/2}\{\hedf_c -
    F\}$.}}
\label{tableBiasVar}
\end{minipage}
\begin{minipage}{\textwidth}
\centering
\vspace*{2em}
\begin{tabular}{|r|c c c c|c|}
\hline
\diagbox{$n$}{$t$} & $-3$ & $-2$ & $-1$ & $0$ & AMISE \\
\hline
100 & 0.0326 & 0.0619 & 0.1516 & 0.1967 & 0.8248 \\
200 & 0.0334 & 0.0782 & 0.1826 & 0.2191 & 0.9248 \\
500 & 0.0098 & 0.0496 & 0.1611 & 0.2077 & 0.7184 \\
1000 & 0.0030 & 0.0362 & 0.1226 & 0.1916 & 0.5812 \\
$\infty$ & 0.0025 & 0.0270 & 0.0913 & 0.1817 & 0.4231 \\
\hline
\end{tabular}
\vspace{1ex}
\caption{Simulated asymptotic mean squared error, at the points $-3$,
  $-2$, $-1$ and $0$, and asymptotic mean integrated squared error of
  $n^{-1/2}\{\hedf_c - F\}$.}
\label{tableMSEMISE}
\end{minipage}
\end{table}

To check the performance of our proposed estimator, we have conducted
simulations of 1000 runs using samples of sizes 100, 200, 500 and
1000. The distribution function has been estimated at the points $0$,
$-1$, $-2$ and $-3$ (the results for t--values $1$, $2$, and $3$ are
very similar). Table \ref{tableBiasVar} shows the results of the
simulated asymptotic bias and variance of $\hedf_c$, which is
calculated by multiplying the simulated bias by the square-root of
each sample size and multiplying the simulated variance by each sample
size. These quantities are predicted to be stable across sample sizes
by Theorem \ref{ThmFhatMARmodel}, and, therefore, will change only
with the value of $t$. Table \ref{tableMSEMISE} shows the results of
the simulated asymptotic mean squared error (AMSE) and the simulated
asymptotic mean integrated squared error (AMISE), which are calculated
similarly to the simulated asymptotic variance. In addition, we have
calculated the AMSE and AMISE for an infinitely large sample using the
results of Theorem \ref{ThmFhatMARmodel}, which are given by the
figures labeled with sample size $\infty$.

Beginning with Table \ref{tableBiasVar}, we can see the asymptotic
bias in $\hedf_c$ is slightly negative near zero, increases to become
positive when moving away from zero and, finally, decreases toward
zero again when moving into the tails of the distribution. This is in
contrast to the asymptotic variance, which appears to be largest near
zero and only decreases toward zero when moving into the tails of the
distribution. Nevertheless, we can see the values appear reasonably
stable at the larger sample sizes 500 and 1000 as desired. Turning our
attention now to Table \ref{tableMSEMISE}, we can plainly see the
estimator $\hedf_c$ appears to have both AMSE and AMISE values
decreasing toward the respective predicted limiting values (given by
the $\infty$ figures). This indicates the predictions made by Theorem
\ref{ThmFhatMARmodel} are indeed adequate for describing the limiting
behavior of $\hedf_c$. In conclusion we find the complete case
estimator $\hedf_c$ useful and practical for estimating the
distribution of the errors $F$ in the heteroskedastic MAR model.


\section{Appendix} \label{Appendix}
This section is the proof of Theorem 1, which is, in particular,
concerned with data obtained from a full model.
\begin{proof}[Proof of Theorem 1]
To begin, we decompose the stochastic quantity in the absolute
brackets in the left--hand side of the assertion into a sum of three
remainder terms:
\begin{equation*}
R_1(t) = \avj \bigg\{\1\big[\hve_j \leq t\big] 
 - E\bigg[F\bigg(t + \frac{\rhat(X) - r(X)}{\sigma(X)}
 + t\frac{\shat(X) - \sigma(X)}{\sigma(X)}\bigg) 
 \,\bigg|\, \mathbb{D}\bigg]
 - \1\big[e_j \leq t\big] + F(t)\bigg\},
\end{equation*}
\begin{align*}
R_2(t) &= E\bigg[F\bigg(t + \frac{\rhat(X) - r(X)}{\sigma(X)}
 + t\frac{\shat(X) - \sigma(X)}{\sigma(X)}\bigg) 
 \,\bigg|\, \mathbb{D}\bigg] - F(t) \\
&\quad - f(t)\int_{[0,\,1]^m}\, \frac{\rhat(x) - r(x)}{\sigma(x)}g(x)\,dx
 - tf(t)\int_{[0,\,1]^m}\, \frac{\shat(x) - \sigma(x)}{\sigma(x)}g(x)\,dx
\end{align*}
and
\begin{align*}
R_3(t) &= f(t)\bigg(\int_{[0,\,1]^m}\, 
 \frac{\rhat(x) - r(x)}{\sigma(x)}g(x)\,dx - \avj e_j\bigg) \\
&\quad + tf(t)\bigg(\int_{[0,\,1]^m}\, 
 \frac{\shat(x) - \sigma(x)}{\sigma(x)}g(x)\,dx 
 - \avj\frac{e_j^2 - 1}{2}\bigg).
\end{align*}
The proof will be concluded once we have shown $\sup_{t \in \R}
|R_i(t)| = \opn$ for each $i = 1, 2, 3$.

To show $\sup_{t \in \R} |R_1(t)| = \opn$, we will proceed similarly
as in the proof of Theorem 2.1 of Neumeyer and Van Keilegom (2010),
and we refer the reader to that paper for further details. The main
difference between the proof techniques lies in the details concerning
the estimators $\rhat$ and $\shat$, which require us to use an
approximation argument the previous authors can avoid. We begin this
argument by noting analogous conclusions of Corollary 1 hold for the
local polynomial estimators $\rhat$ and $\rhat_2$, which follows from
the discussion in Remark 1. This means there are random functions
$\ahat_1$ (associated with $\rhat$) and $\ahat_2$ (associated with
$\rhat_2$) that satisfy $P(\ahat_1 \in H_1(m,\alpha)) \to 1$ and
$P(\ahat_2 \in H_1(m,\alpha)) \to 1$, as $n \to \infty$, for some
$\alpha > 0$, $\sup_{x \in [0,\,1]^m} |\rhat(x) - r(x) - \ahat_1(x)| =
\opn$ and $\sup_{x \in [0,\,1]^m} |\rhat_2(x) - r_2(x) - \ahat_2(x)| =
\opn$. In their Lemma A.3, Neumeyer and Van Keilegom (2010) show a
class of functions similar to
\begin{align*}
\Fclass = \bigg\{ (x,\,z) \mapsto & \1\bigg[z \leq t 
 + \bigg\{1 - t\frac{r(x)}{\sigma(x)}\bigg\}\frac{a_1(x)}{\sigma(x)} 
 + \frac{t}{2}\frac{a_2(x)}{\sigma^2(x)}\bigg] \\
& - E\bigg[ F\bigg(t 
 + \bigg\{1 - t\frac{r(X)}{\sigma(X)}\bigg\}\frac{a_1(X)}{\sigma(X)} 
 + \frac{t}{2}\frac{a_2(X)}{\sigma^2(X)}\bigg) \bigg]
 \,:\, t \in \R,\, a_1,a_2 \in H_1(m,\alpha)\bigg\}
\end{align*}
is $G\times F$--Donsker, and, since our argument that $\Fclass$ is
also $G\times F$--Donsker is very similar, it is omitted. It then
follows by Corollary 2.3.12 in van der Vaart and Wellner (1996) for
the stochastic equicontinuity condition for empirical processes
ranging over $\mathfrak{F}$ to hold, i.e.\ writing $\fmem$ for the
form of the map in the definition of $\Fclass$ above, we have, for any
$\epsilon > 0$,
\begin{equation} \label{equicon}
\lim_{\alpha \downarrow 0} \limsup_{n \to \infty} P\bigg(
 \sup_{\fmem_1,\fmem_2 \in \Fclass \,:\, \Var[\fmem_1(X,e) - \fmem_2(X,e)] \leq \alpha}
 n^{-1/2}\bigg|\sj \Big\{\fmem_1\big(X_j,e_j\big) 
 - \fmem_2\big(X_j,e_j\big)\Big\}\bigg| > \epsilon
 \bigg) = 0.
\end{equation}

We will now use equicontinuity of empirical processes indexed by
$\Fclass$ to finish proving the assertion. In what follows we may
assume that $\ahat_1$ and $\ahat_2$ belong to $H_1(m,\alpha)$, which
we have already shown is an event with probability tending to one as
$n$ increases. Hence, we have $\fmem_{t,\ahat_1,\ahat_2} \in \Fclass$,
where now the expected value is conditional on the data $\Dset =
\{(X_1,Y_1),\ldots,(X_n,Y_n)\}$ and we have chosen $\ahat_1$ for $a_1$
and $\ahat_2$ for $a_2$ in the definition of $\Fclass$. We also have
$\fmem_{t,0,0} \in \Fclass$, which now corresponds with choosing the
zero function for both of $a_1$ and $a_2$ in the definition of
$\Fclass$ above. Inspecting page 961 of the proof of Lemma 1 of
M\"uller, Schick and Wefelmeyer (2009) shows for their situation
$\sup_{x \in [0,\,1]^m}|\ahat(x)| = \op$, which continues to hold in
the present situation, i.e.\ both $\ahat_1$ and $\ahat_2$ satisfy
\begin{equation*}
\sup_{x \in [0,\,1]^m} \big|\ahat_1(x)\big| = \op 
\quad\text{and}\quad 
\sup_{x \in [0,\,1]^m}\big|\ahat_2(x)\big| = \op.
\end{equation*}

In order to apply \eqref{equicon}, we need to consider the variation
of the difference $\fmem_{t,\ahat_1,\ahat_2}(X,e) -
\fmem_{t,0,0}(X,e)$ and find that it is asymptotically negligible. To
check the variance condition beneath the supremum in \eqref{equicon}
is satisfied, we calculate $\Var[\fmem_{t,\ahat_1,\ahat_2}(X,e) -
\fmem_{t,0,0}(X,e) \,|\, \Dset]$ and find it is equal to
\begin{align*}
&E\Bigg[ \bigg\{ \1\bigg[e \leq t 
 + \bigg\{1 - t\frac{r(X)}{\sigma(X)}\bigg\}
 \frac{\ahat_1(X)}{\sigma(X)}
 + \frac{t}{2}\frac{\ahat_2(X)}{\sigma^2(X)} \bigg] \\
&\quad\phantom{E\bigg[ \bigg\}}
 - E\bigg[ F\bigg(t 
 + \bigg\{1 - t\frac{r(X)}{\sigma(X)}\bigg\}
 \frac{\ahat_1(X)}{\sigma(X)}
 + \frac{t}{2}\frac{\ahat_2(X)}{\sigma^2(X)}\bigg)
 \,\bigg|\,\Dset\bigg] \bigg\}^2 
 \,\Bigg|\, \Dset \Bigg] \\
&\quad - 2E\bigg[\bigg\{\1\bigg[e \leq t 
 + \bigg\{1 - t\frac{r(X)}{\sigma(X)}\bigg\}
 \frac{\ahat_1(X)}{\sigma(X)}
 + \frac{t}{2}\frac{\ahat_2(X)}{\sigma^2(X)}\bigg] \\
&\qquad\phantom{2E\bigg[\bigg\{}
 - E\bigg[ F\bigg(t 
 + \bigg\{1 - t\frac{r(X)}{\sigma(X)}\bigg\}
 \frac{\ahat_1(X)}{\sigma(X)}
 + \frac{t}{2}\frac{\ahat_2(X)}{\sigma^2(X)}\bigg)
 \,\bigg|\,\Dset\bigg] \bigg\}
 \times\bigg\{\1[e \leq t] - F(t)\bigg\}\,\bigg|\, \Dset\bigg] \\
&\quad + E\Big[\big\{\1[e \leq t] - F(t)\big\}^2\Big] \\
&= E\bigg[F\bigg(t + \bigg\{1 - t\frac{r(X)}{\sigma(X)}\bigg\}
 \frac{\ahat_1(X)}{\sigma(X)}
 + \frac{t}{2}\frac{\ahat_2(X)}{\sigma^2(X)}\bigg)
 \,\bigg|\, \Dset\bigg] \\
&\quad - E^2\bigg[ F\bigg(t + \bigg\{
 1 - t\frac{r(X)}{\sigma(X)} \bigg\}\frac{\ahat_1(X)}{\sigma(X)}
 + \frac{t}{2}\frac{\ahat_2(X)}{\sigma^2(X)}\bigg)
 \,\bigg|\, \Dset \bigg] \\
&\quad + F(t) - F^2(t) -2E\bigg[ F\bigg( \min\bigg\{ t,\,t
 + \bigg\{1 - t\frac{r(X)}{\sigma(X)}\bigg\}
 \frac{\ahat_1(X)}{\sigma(X)}
 + \frac{t}{2}\frac{\ahat_2(X)}{\sigma^2(X)} \bigg\} \bigg)
 \,\bigg|\, \Dset \bigg] \\
&\quad + 2F(t)E\bigg[ F\bigg( t
 + \bigg\{1 - t\frac{r(X)}{\sigma(X)}\bigg\}
 \frac{\ahat_1(X)}{\sigma(X)}
 + \frac{t}{2}\frac{\ahat_2(X)}{\sigma^2(X)}\bigg)
 \,\bigg|\, \Dset \bigg] \\
&= E\bigg[F\bigg( t
 + \bigg\{1 - t\frac{r(X)}{\sigma(X)}\bigg\}
 \frac{\ahat_1(X)}{\sigma(X)}
 + \frac{t}{2}\frac{\ahat_2(X)}{\sigma^2(X)}\bigg) \\
&\qquad\phantom{E\bigg[}
 - F\bigg( \min\bigg\{ t,\,t
 + \bigg\{1 - t\frac{r(X)}{\sigma(X)}\bigg\}
 \frac{\ahat_1(X)}{\sigma(X)}
 + \frac{t}{2}\frac{\ahat_2(X)}{\sigma^2(X)} \bigg\} \bigg) 
 \,\bigg|\, \Dset \bigg] \\
&\quad + E\bigg[ F(t) - F\bigg( \min\bigg\{ t,\,t 
 + \bigg\{1 - t\frac{r(X)}{\sigma(X)}\bigg\}
 \frac{\ahat_1(X)}{\sigma(X)}
 + \frac{t}{2}\frac{\ahat_2(X)}{\sigma^2(X)} \bigg\} \bigg) 
 \,\bigg|\, \Dset \bigg] \\
&\quad + \bigg\{ E\bigg[ F(t) - F\bigg( t
 + \bigg\{1 - t\frac{r(X)}{\sigma(X)}\bigg\}
 \frac{\ahat_1(X)}{\sigma(X)}
 + \frac{t}{2}\frac{\ahat_2(X)}{\sigma^2(X)}\bigg)
 \,\bigg|\, \Dset \bigg] \bigg\} \\
&\qquad\times E\bigg[ F\bigg( t
 + \bigg\{1 - t\frac{r(X)}{\sigma(X)}\bigg\}
 \frac{\ahat_1(X)}{\sigma(X)}
 + \frac{t}{2}\frac{\ahat_2(X)}{\sigma^2(X)} \bigg)
 \,\bigg|\,\Dset \bigg] \\
&\quad + \bigg\{ E\bigg[ F\bigg( t
 + \bigg\{1 - t\frac{r(X)}{\sigma(X)}\bigg\}
 \frac{\ahat_1(X)}{\sigma(X)}
 + \frac{t}{2}\frac{\ahat_2(X)}{\sigma^2(X)}\bigg) - F(t)
 \,\bigg|\, \Dset \bigg] \bigg\} \times F(t) \\
&= E\bigg[ F\bigg( \max\bigg\{ t,\,t
 + \bigg\{1 - t\frac{r(X)}{\sigma(X)}\bigg\}
 \frac{\ahat_1(X)}{\sigma(X)}
 + \frac{t}{2}\frac{\ahat_2(X)}{\sigma^2(X)}\bigg\} \bigg) \\
&\qquad\phantom{E\bigg[}
 - F\bigg( \min\bigg\{ t,\,t
 + \bigg\{1 - t\frac{r(X)}{\sigma(X)}\bigg\}
 \frac{\ahat_1(X)}{\sigma(X)}
 + \frac{t}{2}\frac{\ahat_2(X)}{\sigma^2(X)} \bigg\} \bigg)
 \,\bigg|\, \Dset \bigg] \\
&\quad - E^2\bigg[F\bigg(\max\bigg\{t,\,t 
 + \bigg\{1 - t\frac{r(X)}{\sigma(X)}\bigg\}
 \frac{\ahat_1(X)}{\sigma(X)}
 + \frac{t}{2}\frac{\ahat_2(X)}{\sigma^2(X)} \bigg\} \bigg) \\
&\qquad\phantom{-E^2\bigg[}
 - F\bigg( \min\bigg\{ t,\,t
 + \bigg\{1 - t\frac{r(X)}{\sigma(X)}\bigg\}
 \frac{\ahat_1(X)}{\sigma(X)}
 + \frac{t}{2}\frac{\ahat_2(X)}{\sigma^2(X)} \bigg\} \bigg)
 \,\bigg|\, \Dset \bigg].
\end{align*}
Therefore, we find
\begin{align*}
&\sup_{t \in \R} \Var\big[ \fmem_{t,\ahat_1,\ahat_2}(X,e) 
 - \fmem_{t,0,0}(X,e) \,|\, \Dset \big] \\
&\leq E\bigg[ F\bigg( \max\bigg\{ t,\,t
 + \bigg\{1 - t\frac{r(X)}{\sigma(X)}\bigg\}
 \frac{\ahat_1(X)}{\sigma(X)}
 + \frac{t}{2}\frac{\ahat_2(X)}{\sigma^2(X)} \bigg\} \bigg) \\
&\qquad\phantom{E\bigg[}
 - F\bigg( \min\bigg\{ t,\,t
 + \bigg\{1 - t\frac{r(X)}{\sigma(X)}\bigg\}
 \frac{\ahat_1(X)}{\sigma(X)}
 + \frac{t}{2}\frac{\ahat_2(X)}{\sigma^2(X)} \bigg\} \bigg)
 \,\bigg|\, \Dset \bigg] \\
&\leq \bigg( \sup_{t \in \R} f(t)
 \bigg[\inf_{x \in [0,\,1]^m} \sigma(x)\bigg]^{-1}
 + \sup_{t \in \R} \big|tf(t)\big|
 \sup_{x \in [0,\,1]^m} |r(x)|
 \bigg[\inf_{x \in [0,\,1]^m} \sigma(x)\bigg]^{-2} \bigg)
 \sup_{x \in [0,\,1]^m} \big|\ahat_1(x)\big| \\
&\quad + \frac12 \sup_{t \in \R} \big|tf(t)\big|
 \bigg[\inf_{x \in [0,\,1]^m} \sigma(x)\bigg]^{-2}
 \sup_{x \in [0,\,1]^m} \big|\ahat_2(x)\big|.
\end{align*}
It then follows that $\Var[\fmem_{t,\ahat_1,\ahat_2}(X,e) -
\fmem_{t,0,0}(X,e)\,|\,\mathbb{D}] = \op$, $t \in \R$, from the
results above, i.e.\ the variance is asymptotically negligible. Hence,
have asymptotic equicontinuity for empirical processes indexed by the
elements of $\Fclass$ corresponding to the choices $\ahat_1$ and
$\ahat_2$, for $a_1$ and $a_2$, and the zero function, in place of
both $a_1$ and $a_2$. This implies
\begin{equation} \label{R1remainder1}
\sup_{t \in \R}\bigg|\avj \Big\{
 \fmem_{t,\ahat_1,\ahat_2}\big(X_j,e_j\big) 
 - \fmem_{t,0,0}\big(X_j,e_j\big)\Big\}\bigg| = \opn.
\end{equation}
Note, the statement above continues to hold using
$\fmem_{n,t,\ahat_1,\ahat_2}$ in place of $\fmem_{t,\ahat_1,\ahat_2}$,
where $\{\fmem_{n,t,\ahat_1,\ahat_2}\}_{n \geq 1}$ is any sequence of
functions from $\Fclass$ converging to the same limit as
$\fmem_{t,\ahat_1,\ahat_2}$ (i.e.\ $\fmem_{t,0,0}$), as $n$ increases,
that continues to satisfy the variation condition above, which will be
an important observation for the arguments that follow.

We can bound $\sup_{t \in \R} |R_1(t)|$ by the sum of the left--hand
side of \eqref{R1remainder1} and 
\begin{equation} \label{R1remainder2}
\sup_{t \in \R}
 \bigg|\avj \Big\{\chi_t\big(X_j,e_j\big) 
 - \fmem_{t,\ahat_1,\ahat_2}\big(X_j,e_j\big)\Big\} \bigg|,
\end{equation}
where $\chi_t(X_j,e_j)$, $j=1,\ldots,n$ and $t \in \R$, is equal to
\begin{equation*}
\1\bigg[e_j \leq t + 
 \frac{\rhat(X_j) - r(X_j)}{\sigma(X_j)}
 + t\frac{\shat(X_j) - \sigma(X_j)}{\sigma(X_j)}\bigg]
 - E\bigg[F\bigg(t + 
 \frac{\rhat(X) - r(X)}{\sigma(X)}
 + t\frac{\shat(X) - \sigma(X)}{\sigma(X)}\bigg)
 \,\bigg|\, \mathbb{D} \bigg].
\end{equation*}
To continue, we will require some auxiliary results. By a direct
application of Theorem 6 of Masry (1996), we have $\sup_{x \in
  [0,\,1]^m} \{\rhat(x) - r(x)\}^2 = o(n^{-1/2})$, $\sup_{x \in
  [0,\,1]^m} \{\rhat_2(x) - r_2(x)\}^2 = o(n^{-1/2})$ and $\sup_{x \in
  [0,\,1]^m}|\rhat(x) - r(x)||\rhat_2(x) - r_2(x)| = o(n^{-1/2})$,
almost surely. Now write
\begin{equation*}
\shat(x) - \sigma(x) = \frac{\shat^2(x) - \sigma^2(x)}{2\sigma(x)}
 -\frac{\{\shat(x) - \sigma(x)\}^2}{2\sigma(x)}
\end{equation*}
and
\begin{equation*}
\rhat^2(x) - r^2(x) = 2r(x)\big\{\rhat(x) - r(x)\big\} 
 + \big\{\rhat(x) - r(x)\big\}^2.
\end{equation*}
Together these results imply
\begin{align*}
\shat(x) - \sigma(x) &= \frac{\rhat_2(x) - r_2(x)}{2\sigma(x)}
 - \frac{r(x)}{\sigma(x)}\big\{\rhat(x) - r(x)\big\}
 - \frac{\{\rhat(x) - r(x)\}^2}{2\sigma(x)}
 - \frac{\{\shat(x) - \sigma(x)\}^2}{2\sigma(x)}
\end{align*}
and
\begin{align*}
\big\{\shat(x) - \sigma(x)\big\}^2 &=
 \frac{\{\rhat_2(x) - r(x)\}^2}{4\sigma^2(x)}
 - \frac{r(x)}{\sigma^2(x)}\big\{\rhat_2(x) - r_2(x)\big\}
 \big\{\rhat(x) - r(x)\big\} \\
&\quad - \frac{\{\rhat_2(x) - r_2(x)\}
 \{\rhat(x) - r(x)\}^2}{2\sigma^2(x)}
 - \frac{\{\rhat_2(x) - r_2(x)\}
 \{\shat(x) - \sigma(x)\}^2}{2\sigma^2(x)} \\
&\quad + \frac{r^2(x)}{\sigma^2(x)}\big\{\rhat(x) - r(x)\big\}^2
 + \frac{r(x)}{\sigma^2(x)}\big\{\rhat(x) - r(x)\big\}^3 \\
&\quad + \frac{r(x)}{\sigma^2(x)}\big\{\rhat(x) - r(x)\big\}
 \big\{\shat(x) - \sigma(x)\big\}^2
 + \frac{\{\rhat(x) - r(x)\}^4}{4\sigma^2(x)} \\
&\quad + \frac{\{\rhat(x) - r(x)\}^2
 \{\shat(x) - \sigma(x)\}^2}{2\sigma^2(x)}
 + \frac{\{\shat(x) - \sigma(x)\}^4}{4\sigma^2(x)}.
\end{align*}
Combining the last statement with the results above, we find
$\sup_{x \in [0,\,1]^m} \{\shat(x) - \sigma(x)\}^2 = o(n^{-1/2})$,
almost surely. Using the definitions of $\rhat$ and $\shat$, we find
$t + \{\rhat(x) - r(x)\}/\sigma(x) + t\{\shat(x) -
\sigma(x)\}/\sigma(x)$ is equal to 
\begin{gather*}
t\bigg\{1 + \frac{\rhat_2(x) - r_2(x) - \ahat_2(x)}{2\sigma^2(x)}
 - \frac{r(x)}{\sigma(x)}
 \frac{\rhat(x) - r(x) - \ahat_1(x)}{\sigma(x)}
 - \frac{\{\rhat(x) - r(x)\}^2}{2\sigma^2(x)}
 - \frac{\{\shat(x) - \sigma(x)\}^2}{2\sigma^2(x)}\bigg\} \\
 + \frac{\rhat(x) - r(x) - \ahat_1(x)}{\sigma(x)} 
 + \bigg\{1 - t\frac{r(x)}{\sigma(x)}\bigg\}
 \frac{\ahat_1(x)}{\sigma(x)}
 + \frac{t}{2}\frac{\ahat_2(x)}{\sigma^2(x)}.
\end{gather*}
This means we can appropriately choose random sequences $\{u_n\}_{n
  \geq 1}$ and $\{v_n\}_{n \geq 1}$, where
\begin{align*}
u_n &= 4\bigg[\inf_{x \in [0,\,1]^m} \sigma(x)\bigg]^{-2} 
 \max\bigg\{
 \frac{n^{1/2}}{2} \sup_{x \in [0,\,1]^m}
 \big|\rhat_2(x) - r(x) - \ahat_2(x)\big|, \\
&\quad\phantom{= 4\bigg[\inf_{x \in [0,\,1]^m} \sigma(x)\bigg]^{-2}
 \max\bigg\{}
 n^{1/2} \sup_{x \in [0,\,1]^m}\big|r(x)\big|
 \sup_{x \in [0,\,1]^m}\big|\rhat(x) - r(x) - \ahat_1(x)\big|, \\
&\quad\phantom{= 4\bigg[\inf_{x \in [0,\,1]^m} \sigma(x)\bigg]^{-2}
 \max\bigg\{}
 \frac{n^{1/2}}{2} \sup_{x \in [0,\,1]^m}
 \big\{\rhat(x) - r(x)\big\}^2,\,
 \frac{n^{1/2}}{2} \sup_{x \in [0,\,1]^m}
 \big\{\shat(x) - \sigma(x)\big\}^2\bigg\} \\
&= \op
\end{align*}
and
\begin{equation*}
v_n = n^{1/2} \bigg[\inf_{x \in [0,\,1]^m} \sigma(x)\bigg]^{-1}
 \sup_{x \in [0,\,1]^m} \big|\rhat(x) - r(x) - \ahat_1(x)\big| \\
 = \op,
\end{equation*}
which depend only on $\Dset$. We then define function sequences
$\fmem_{n,t,\ahat_1,\ahat_2}^{++}$,
$\fmem_{n,t,\ahat_1,\ahat_2}^{+-}$, $\fmem_{n,t,\ahat_1,\ahat_2}^{-+}$
and $\fmem_{n,t,\ahat_1,\ahat_2}^{--}$ from $\Fclass$ as
follows. Define the sequences $\fmem_{n,t,\ahat_1,\ahat_2}^{++}$ and
$\fmem_{n,t,\ahat_1,\ahat_2}^{--}$ as
\begin{align*}
\fmem_{n,t,\ahat_1,\ahat_2}^{++}(x,z) &= 
 \1\bigg[ z \leq t\big\{1 + n^{-1/2}u_n\big\} + n^{-1/2}v_n 
 + \bigg\{1 - t\frac{r(x)}{\sigma(x)}\bigg\}
 \frac{\ahat_1(x)}{\sigma(x)}
 + \frac{t}{2}\frac{\ahat_2(x)}{\sigma^2(x)} \bigg] \\
&\quad - E\bigg[ F\bigg( t\big\{1 + n^{-1/2}u_n\big\} + n^{-1/2}v_n
 + \bigg\{1 - t\frac{r(X)}{\sigma(X)}\bigg\}
 \frac{\ahat_1(X)}{\sigma(X)}
 + \frac{t}{2}\frac{\ahat_2(X)}{\sigma^2(X)} \bigg)
 \,\bigg|\, \Dset \bigg]
\end{align*}
and
\begin{align*}
\fmem_{n,t,\ahat_1,\ahat_2}^{--}(x,z) &=
 \1\bigg[ z \leq t\big\{1 - n^{-1/2}u_n\big\} - n^{-1/2}v_n
 + \bigg\{1 - t\frac{r(x)}{\sigma(x)}\bigg\}
 \frac{\ahat_1(x)}{\sigma(x)}
 + \frac{t}{2}\frac{\ahat_2(x)}{\sigma^2(x)} \bigg] \\
&\quad - E\bigg[ F\bigg( t\big\{1 - n^{-1/2}u_n\big\} - n^{-1/2}v_n
 + \bigg\{1 - t\frac{r(X)}{\sigma(X)}\bigg\}
 \frac{\ahat_1(X)}{\sigma(X)}
 + \frac{t}{2}\frac{\ahat_2(X)}{\sigma^2(X)} \bigg)
 \,\bigg|\, \Dset \bigg].
\end{align*}
The remaining two sequences are defined similarly to those above. We
will now consider the case $t \in [0,\,\infty)$. On this region, we
can bound \eqref{R1remainder2} by a sum of three terms:
\begin{equation} \label{equiconupper}
\sup_{t \in [0,\,\infty)} \bigg|\avj \Big\{
 \fmem_{n,t,\ahat_1,\ahat_2}^{++}\big(X_j,e_j\big)
 - \fmem_{n,t,\ahat_1,\ahat_2}^{--}\big(X_j,e_j\big) \Big\} \bigg|,
\end{equation}
\begin{align} \label{equiconrem}
\sup_{t \in [0,\,\infty)} \bigg|&
 E\bigg[ F\bigg( t\big\{1 + n^{-1/2}u_n\big\} + n^{-1/2}v_n
 + \bigg\{1 - t\frac{r(X)}{\sigma(X)}\bigg\}
 \frac{\ahat_1(X)}{\sigma(X)}
 + \frac{t}{2}\frac{\ahat_2(X)}{\sigma^2(X)} \bigg)
 \,\bigg|\, \Dset \bigg] \\ \nonumber
& - E\bigg[ F\bigg( t\big\{1 - n^{-1/2}u_n\big\} - n^{-1/2}v_n
 + \bigg\{1 - t\frac{r(X)}{\sigma(X)}\bigg\}
\frac{\ahat_1(X)}{\sigma(X)}
 + \frac{t}{2}\frac{\ahat_2(X)}{\sigma^2(X)} \bigg)
 \,\bigg|\, \Dset \bigg] \bigg|
\end{align}
and 
\begin{align} \label{origrem}
\sup_{t \in [0,\,\infty)} \bigg|
 &E\bigg[ F\bigg( t + \frac{\rhat(X) - r(X)}{\sigma(X)}
 + t\frac{\shat(X) - \sigma(X)}{\sigma(X)} \bigg) \\ \nonumber
&\quad - F\bigg( t + \bigg\{1 - t\frac{r(X)}{\sigma(X)}\bigg\}
 \frac{\ahat_1(X)}{\sigma(X)}
 + \frac{t}{2}\frac{\ahat_2(X)}{\sigma^2(X)} \bigg)
 \,\bigg|\, \Dset \bigg] \bigg|.
\end{align}
Analogous arguments for the case of $t \in (-\infty,\,0)$ lead to a
similar bound, where $\fmem_{n,t,\ahat_1,\ahat_2}^{-+}$ replaces
$\fmem_{n,t,\ahat_1,\ahat_2}^{++}$, $\fmem_{n,t,\ahat_1,\ahat_2}^{+-}$
replaces $\fmem_{n,t,\ahat_1,\ahat_2}^{--}$, \eqref{equiconrem} is
appropriately adjusted and the supremum in each term is restricted to
$(-\infty,\,0)$.

Following the arguments above, we will now specialize the asymptotic
equicontinuity condition above to show \eqref{equiconupper} is
$\opn$. Repeating the calculations above for
$\Var[\fmem_{t,\ahat_1,\ahat_2}(X,e) - \fmem_{t,0,0}(X,e) \,|\,
\Dset]$, we find, now for $t \in \R$,
\begin{align*}
&\sup_{t \in \R} \Var\big[\fmem_{n,t,\ahat_1,\ahat_2}^{++}(X,e) 
 - \fmem_{n,t,\ahat_1,\ahat_2}^{--}(X,e)\,|\,\Dset\big] \\
&\leq E\bigg[ F\bigg( \max\bigg\{ t\big(1 + n^{-1/2}u_n\big)
 + n^{-1/2}v_n + \bigg\{1 - t\frac{r(X)}{\sigma(X)}\bigg\}
 \frac{\ahat_1(X)}{\sigma(X)}
 + \frac{t}{2}\frac{\ahat_2(X)}{\sigma^2(X)}, \\
&\qquad\phantom{E\bigg[F\bigg(\max\bigg\{} 
 t\big(1 - n^{-1/2}u_n\big) - n^{-1/2}v_n
 + \bigg\{1 - t\frac{r(X)}{\sigma(X)}\bigg\}
 \frac{\ahat_1(X)}{\sigma(X)}
 + \frac{t}{2}\frac{\ahat_2(X)}{\sigma^2(X)} \bigg\} \bigg) \\
&\quad\phantom{E\bigg[}
 - F\bigg( \min\bigg\{ t\big(1 + n^{-1/2}u_n\big) + n^{-1/2}v_n
 + \bigg\{1 - t\frac{r(X)}{\sigma(X)}\bigg\}
 \frac{\ahat_1(X)}{\sigma(X)}
 + \frac{t}{2}\frac{\ahat_2(X)}{\sigma^2(X)}, \\
&\qquad\phantom{E\bigg[ - F\bigg(\min\bigg\{}
 t\big(1 - n^{-1/2}u_n\big) - n^{-1/2}v_n
 + \bigg\{1 - t\frac{r(X)}{\sigma(X)}\bigg\}
\frac{\ahat_1(X)}{\sigma(X)}
 + \frac{t}{2}\frac{\ahat_2(X)}{\sigma^2(X)} \bigg\} \bigg)
 \,\bigg|\, \Dset\bigg] \\
&\leq 2n^{-1/2} v_n \sup_{t \in \R} f(t)
 + 2n^{-1/2} u_n \sup_{t \in \R} \big|tf(t)\big|.
\end{align*}
This bound is $\opn$ and, therefore, $\op$, which implies the variance
is asymptotically negligible:
$\Var[\fmem_{n,t,\ahat_1,\ahat_2}^{++}(X,e) -
\fmem_{n,t,\ahat_1,\ahat_2}^{--}(X,e)\,|\,\Dset] = \op$, uniformly in
$t \in \R$. Hence, we have asymptotic equicontinuity and it follows
for \eqref{equiconupper} to be $\opn$ as desired by further
restricting $t$ to $[0,\,\infty)$. Continuing, we have
\begin{align*}
&\sup_{t \in \R} \bigg|
 E\bigg[ F\bigg( t\big(1 + n^{-1/2}u_n\big) + n^{-1/2}v_n
 + \bigg\{1 - t\frac{r(X)}{\sigma(X)}\bigg\}
 \frac{\ahat_1(X)}{\sigma(X)}
 + \frac{t}{2}\frac{\ahat_2(X)}{\sigma^2(X)} \bigg) \\
&\phantom{\sup_{t \in \R} \bigg| E\bigg[}
 - F\bigg( t\big(1 - n^{-1/2}u_n\big) - n^{-1/2}v_n
 + \bigg\{1 - t\frac{r(X)}{\sigma(X)}\bigg\}
 \frac{\ahat_1(X)}{\sigma(X)}
 + \frac{t}{2}\frac{\ahat_2(X)}{\sigma^2(X)} \bigg)
 \,\bigg|\, \Dset\bigg] \bigg| \\
&\leq 2n^{-1/2} v_n \sup_{t \in \R} f(t)
 + 2n^{-1/2} u_n \sup_{t \in \R} \big|tf(t)\big|.
\end{align*}
Since $\{u_n\}_{n \geq 1}$ and $\{v_n\}_{n \geq 1}$ are both $\op$, it
follows for the bound above to be $\opn$. This also implies
\eqref{equiconrem} is $\opn$. Now using the identity for $t +
\{\rhat(x) - r(x)\}/\sigma(x) + t\{\shat(x) - \sigma(x)\}/\sigma(x)$,
we have
\begin{align*}
&\sup_{t \in \R} \bigg| E\bigg[ F\bigg( t
 + \frac{\rhat(X) - r(X)}{\sigma(X)}
 + t\frac{\shat(X) - \sigma(X)}{\sigma(X)} \bigg) \\
&\quad\phantom{\sup_{t \in \R} \bigg| E\bigg[}
 - F\bigg( t + \bigg\{1 - t\frac{r(X)}{\sigma(X)}\bigg\}
 \frac{\ahat_1(X)}{\sigma(X)}
 + \frac{t}{2}\frac{\ahat_2(X)}{\sigma^2(X)} \bigg)
 \,\bigg|\, \Dset \bigg] \bigg| \\
&= \sup_{t \in \R} \bigg| E\bigg[ F\bigg( t
 + \bigg\{1 - t\frac{r(X)}{\sigma(X)}\bigg\}
 \frac{\ahat_1(X)}{\sigma(X)}
 + \frac{t}{2}\frac{\ahat_2(X)}{\sigma^2(X)}
 + \frac{\rhat(X) - r(X) - \ahat_1(X)}{\sigma(X)} \\
&\quad\phantom{= \sup_{t \in \R} \bigg| E\bigg[F\bigg(}
 + t\bigg\{
 \frac{\rhat_2(X) - r_2(X) - \ahat_2(X)}{2\sigma^2(X)}
 - \frac{r(X)}{\sigma(X)}\frac{\rhat(X) - r(X)
 - \ahat_1(X)}{\sigma(X)} \\
&\qquad\phantom{= \sup_{t \in \R} \bigg| E\bigg[F\bigg(t\bigg\{}
 - \frac{\{\rhat(X) - r(X)\}^2}{2\sigma^2(X)}
 - \frac{\{\shat(X) - \sigma(X)\}^2}{2\sigma^2(X)}
 \bigg\} \bigg) \\
&\quad\phantom{= \sup_{t \in \R} \bigg| E\bigg[}
 - F\bigg(t + \bigg\{1 - \frac{r(X)}{\sigma(X)}\bigg\}
 \frac{\ahat_1(X)}{\sigma(X)}
 + \frac{t}{2}\frac{\ahat_2(X)}{\sigma^2(X)} \bigg)
 \,\bigg|\, \Dset \bigg] \bigg| \\
&\leq \Bigg( \sup_{t \in \R} f(t)
 \bigg[ \inf_{x \in [0,\,1]^m} \sigma(x)\bigg]^{-1}
 + \sup_{t \in \R} \big|tf(t)\big|
 \bigg[\sup_{x \in [0,\,1]^m} \big|r(x)\big|\bigg]
 \bigg[\inf_{x \in [0,\,1]^m} \sigma(x)\bigg]^{-2} \Bigg) \\
&\quad\times \sup_{x \in [0,\,1]^m}
 \big|\rhat(x) - r(x) - \ahat_1(x)\big| \\
&\quad + \frac12 \sup_{t \in \R} \big|tf(t)\big| 
 \bigg[\inf_{x \in [0,\,1]^m} \sigma(x)\bigg]^{-2}
 \sup_{x \in [0,\,1]^m} \big|\rhat_2(x) - r_2(x) - \ahat_2(x)\big| \\
&\quad + \frac12 \sup_{t \in \R} \big|tf(t)\big|
 \bigg[\inf_{x \in [0,\,1]^m} \sigma(x)\bigg]^{-2}
 \bigg( \sup_{x \in [0,\,1]^m} \big\{\rhat(x) - r(x)\}^2
 + \sup_{x \in [0,\,1]^m} \big\{\shat(x) - \sigma(x)\big\}^2 \bigg).
\end{align*}
Using the results above, this bound is $\opn$, which implies
\eqref{origrem} is $\opn$. Since the same logic can be applied when
$t \in (-\infty,\,0)$, now using the function sequences
$\{\fmem_{n,t,\ahat_1,\ahat_2}^{-+}\}_{n \geq 1}$ and
$\{\fmem_{n,t,\ahat_1,\ahat_2}^{+-}\}_{n \geq 1}$, the analogous
remainder terms to \eqref{equiconupper}, \eqref{equiconrem} and
\eqref{origrem} are all $\opn$ as well. Combining these results shows
$\sup_{t \in \R} |R_1(t)| = \opn$.

From Remark 1, the random functions $\ahat_1$ and $\ahat_2$
additionally satisfy
\begin{equation} \label{ahat12negli}
\int_{[0,\,1]^m}|\ahat_1(x)|^{1 + b}g(x)\,dx = \opn
\quad\text{and}\quad
\int_{[0,\,1]^m}|\ahat_2(x)|^{1 + b}g(x)\,dx = \opn,
\end{equation}
where $b > m/(2s - m)$. Setting $A(x) = \{\rhat(x) - r(x)\}/\sigma(x)$
and $B(x) = \{\shat(x) - \sigma(x)\}/\sigma(x)$, we can then bound
$\sup_{t \in \R} |R_2(t)|$ by a sum of three terms:
\begin{equation} \label{Arem}
\sup_{t \in \R} \Big|E\big[F\big(t + A(X) + B(X)t\big) 
 - F\big(t + B(X)t\big) - A(X)f\big(t + B(X)t\big) 
 \,\big|\,\mathbb{D}\big]\Big|,
\end{equation}
\begin{equation} \label{Brem}
\sup_{t \in \R} \Big|E\big[F\big(t + B(X)t\big) - F(t) - B(X)tf(t)
 \,\big|\,\mathbb{D}\big]\Big|
\end{equation}
and
\begin{equation} \label{ABrem}
\sup_{t \in \R} \Big|E\big[A(X)\big\{f\big(t + B(X)t\big) - f(t)\big\}
 \,\big|\,\mathbb{D}\big]\Big|.
\end{equation}
To continue, we will require an additional result. Setting $x =
\min\{t,\,t + B(X)t\}$ and $y = \max\{t,\,t + B(X)t\}$, we have $0
\leq y - x = \max\{-B(X)t,\,B(X)t\} = |B(X)t|$, and we find, almost
surely,
\begin{align*}
&\big(1 + x^2\big)\big|f(y) - f(x)\big| \\
&\leq \big|\big(1 + y^2\big)f(y) - \big(1 + x^2\big)f(x)\big| \\
&= \bigg|\int_x^y\,\big\{\big(1 + v^2\big)f'(v)\,dv 
 + \int_x^y\,2vf(v)\,dv\bigg| \\
&\leq \int_x^y\,\big(1 + v^2\big)^{1/2}f^{1/2}(v)
 \big(1 + v^2\big)^{1/2}\bigg|\frac{f'(v)}{f(v)}\bigg|f^{1/2}(v)\,dv
 + 2\int_x^y\,\big|vf(v)\big|\,dv \\
&\leq \bigg\{\int_x^{x + y - x}\,\big(1 + v^2\big)f(v)\,dv\bigg\}^{1/2}
 \bigg\{ \int_x^{x + y - x}\,\big(1 + v^2\big)
 \bigg(\frac{f'(v)}{f(v)}\bigg)^2
 f(v)\,dv \bigg\}^{1/2} \\
&\quad + 2\int_x^{x + y - x}\,\big|vf(v)\big|\,dv \\
&\leq \bigg\{\int_0^1 \big(1 + \{x + s(y - x)\}^2\big)
 f\big(x + s(y - x)\big)\,ds\bigg\}^{1/2} \\
&\quad\times \bigg\{\int_0^1 \big(1 + \{x + s(y - x)\}^2\big)
 \bigg(\frac{f'(x + s(y - x))}{f(x + s(y - x))}\bigg)^2
 f\big(x + s(y - x)\big)\,ds\bigg\}^{1/2} \times \big|B(X)t\big| \\
&\quad + 2\bigg(\sup_{t \in \R} \big|tf(t)\big|\bigg)
 \big|B(X)t\big| \\
&\leq \Bigg( \bigg\{ 2\int_{-\infty}^{\infty}
 \big(1 + v^2\big)\bigg(\frac{f'(v)}{f(v)}\bigg)^2\,F(dv)
 \bigg\}^{1/2}
 + 2\sup_{t \in \R} \big|tf(t)\big| \Bigg)\big|B(X)t\big|.
\end{align*}
This implies
\begin{align*}
&\big|f\big(t + B(X)t\big) - f(t)\big| \\ 
&= \big|f\big(\max\big\{t,\,t + B(X)t\big\}\big) 
 - f\big(\min\big\{t,\,t + B(X)t\big\}\big)\big| \\
&\leq \Bigg( \bigg\{ 2\int_{-\infty}^{\infty}
 \big(1 + v^2\big)\bigg(\frac{f'(v)}{f(v)}\bigg)^2\,F(dv)
 \bigg\}^{1/2}
 + 2\sup_{t \in \R} \big|tf(t)\big| \Bigg)
 \frac{|t|}{1 + t^2\min\{1,\,(1 + B(X))^2\}} 
 \big|B(X)\big| \\
&\leq \Bigg( \bigg\{ 2\int_{-\infty}^{\infty}
 \big(1 + v^2\big)\bigg(\frac{f'(v)}{f(v)}\bigg)^2\,F(dv)
 \bigg\}^{1/2}
 + 2\sup_{t \in \R} \big|tf(t)\big| \Bigg)
 \frac{|t|}{1 + t^2\min\{1,\,(1 + L)^2\}} \big|B(X)\big|,
\end{align*}
almost surely, where $L = \inf_{x \in [0,\,1]^m} B(x)$. We have
already shown that $\sup_{x \in   [0,\,1]^m} |B(x)| = o(1)$ above,
almost surely, and, for large enough $n$, $-1 < \inf_{x \in [0,\,1]^m}
B(x) = L$. Hence, for large enough $n$, the bound above is finite.

It then follows for the map $t \mapsto f(t)$ to be H\"older with
exponent $b$. Writing $K_{f,b}$ for the H\"older constant for $f$ with
exponent $b$, we have
\begin{align*}
&\sup_{t \in \R} \big| E\big[F\big(t + A(X) + B(X)t\big) 
 - F\big(t + B(X)t\big) - A(X)f\big(t + B(X)t\big)
 \,\big|\, \Dset\big] \big| \\
&\leq \sup_{t \in \R} E\bigg[|A(X)|\int_0^1\,
 \big|f\big(t + sA(X) + B(X)t\big) - f\big(t + B(X)t\big)\big|\,ds
 \,\bigg|\, \Dset\bigg] \\
&\leq K_{f,b} \int_0^1\,|s|^b\,ds\, E\big[|A(X)|^{1 + b} \,\big|\,
 \Dset\big] \\
&= C_1 E\big[|A(X)|^{1 + b} \,\big|\, \Dset\big],
\end{align*}
choosing $C_1 = K_{f,b}/\{1 + b\}$. Observing that for any real
numbers $x$ and $y$, and $0 < c$, that $|x + y|^{1 + c} \leq
2^c(|x|^{1+c} + |y|^{1 + c})$, we find
\begin{align*}
&E\big[|A(X)|^{1 + b} \,\big|\, \Dset\big] \\
&= E\bigg[\bigg|\frac{\ahat_1(X)}{\sigma(X)} 
 + \frac{\rhat(X) - r(x) - \ahat_1(X)}{\sigma(X)}\bigg|^{1 + b}
 \,\bigg|\, \Dset\bigg] \\
&\leq 2^b\bigg[\inf_{x \in [0,\,1]^m} \sigma(x)\bigg]^{1 + b}
 \Bigg( \int_{[0,\,1]^m} \big|\ahat_1(x)\big|^{1 + b}g(x)\,dx
 + \bigg[\sup_{x \in [0,\,1]^m} 
 \big|\rhat(x) - r(x) - \ahat_1(x)\big|\bigg]^{1 + b} \Bigg).
\end{align*}
Using the results above, this bound is $\opn$, and, therefore,
\eqref{Arem} is $\opn$.

Using the same procedure as above, we find
\begin{align*}
&\sup_{t \in \R} \big| E\big[F\big(t + B(X)t\big) - F(t) - B(X)tf(t)
 \,\big|\, \Dset\big] \big| \\
&= \sup_{t \in \R} \bigg| E\bigg[B(X)t
 \int_0^1 \big\{f\big(t + sB(X)t\big) - f(t)\big\}\,ds
 \,\big|\, \Dset\bigg] \bigg| \\
&\leq \sup_{t \in \R} E\bigg[\big|B(X)\big| 
 \int_0^1 |t|\big|f\big(t + sB(X)t\big) - f(t)\big|\,ds
 \,\bigg|\, \Dset\bigg] \\
&\leq \bigg(\bigg\{2\int_{-\infty}^{\infty}
 \big(1 + v^2\big)\bigg(\frac{f'(v)}{f(v)}\bigg)^2\,F(dv)\bigg\}^{1/2}
 + 2\sup_{t \in \R} \big|tf(t)\big|\bigg) \\
&\quad\times \bigg(\sup_{t \in \R} \sup_{s \in [0,\,1]} 
 \frac{st^2}{1 + t^2\min\{1,\,(1 + sL)^2\}}\bigg)\times
 E\big[B^2(X) \,\big|\, \Dset\big] \\
&\leq \bigg(\bigg\{2\int_{-\infty}^{\infty}
 \big(1 + v^2\big)\bigg(\frac{f'(v)}{f(v)}\bigg)^2\,F(dv)\bigg\}^{1/2}
 + 2\sup_{t \in \R} \big|tf(t)\big|\bigg) \\
&\quad\times
 \bigg(\sup_{t \in \R} 
 \frac{t^2}{1 + t^2\min\{1,\,(1 + L)^2\}}\bigg)
 E\big[B^2(X) \,\big|\, \Dset\big] \\
&\leq C_2 E\big[B^2(X) \,\big|\, \Dset\big],
\end{align*}
for large enough $n$, choosing $C_2$ proportional to 
\begin{equation*}
\Bigg(\bigg\{2\int_{-\infty}^{\infty}
 \big(1 + v^2\big)\bigg(\frac{f'(v)}{f(v)}\bigg)^2\,F(dv)\bigg\}^{1/2}
 + 2\sup_{t \in \R} \big|tf(t)\big|\Bigg) \Bigg(\sup_{t \in \R} 
 \frac{t^2}{1 + t^2\min\{1,\,(1 + L)^2\}}\Bigg).
\end{equation*}
Continuing, $E[B^2(X)\,|\,\Dset]$ is equal to
\begin{align*}
&E\Bigg[\bigg(\frac{\ahat_2(X)}{2\sigma^2(X)}
 - \frac{r(X)}{\sigma(X)}\frac{\ahat_1(X)}{\sigma(X)}
 + \frac{\rhat_2(X) - r_2(X) - \ahat_2(X)}{2\sigma^2(X)}
 - \frac{r(X)}{\sigma(X)}
 \frac{\rhat(X) - r(X) - \ahat_1(X)}{\sigma(X)} \\
&\phantom{E\Bigg[\bigg(}
 - \frac{\{\rhat(X) - r(X)\}^2}{2\sigma^2(X)}
 - \frac{\{\shat(X) - \sigma(X)\}^2}{2\sigma^2(X)}\bigg)^2
 \,\Bigg|\, \Dset \Bigg],
\end{align*}
which is bounded by
\begin{align*}
&\bigg[\inf_{x \in [0,\,1]^m} \sigma(x)\bigg]^{-4}
 \Bigg( \frac12 \int_{[0,\,1]^m}\big|\ahat_2(x)\big|^2g(x)\,dx
 + 2\bigg[\sup_{x \in [0,\,1]^m}
 \big|\rhat_2(x) - r_2(x) - \ahat_2(x)\big|\bigg]^2 \\
& + 4\bigg[\sup_{x \in [0,\,1]^m} \big|r(x)\big|\bigg]^2 
 \int_{[0,\,1]^m} \big|\ahat_1(x)\big|^2g(x)\,dx
 + 8\bigg[\sup_{x \in [0,\,1]^m}
 \big\{\rhat(x) - r(x)\big\}^2\bigg]^2 \\
& + 8\bigg[\sup_{x \in [0,\,1]^m}
 \big\{\shat(x) - \sigma(x)\big\}^2\bigg]^2
 + 16\bigg[\sup_{x \in [0,\,1]^m} \big|r(x)\big|\bigg]^2
 \bigg[ \sup_{x \in [0,\,1]^m}
 \big|\rhat(x) - r(x) - \ahat_1(x)\big| \bigg]^2
 \Bigg).
\end{align*}
It then follows from the results above and \eqref{ahat12negli} for
this bound to be $\opn$. This implies \eqref{Brem} is $\opn$.

Again, using the procedure above, we find
\begin{align*}
&\sup_{t \in \R} \big| E\big[ A(X)
 \big\{f\big(t + B(X)t\big) - f(t)\big\}
 \,\big|\, \Dset \big] \big| \\
&\leq \sup_{t \in \R} E\big[ \big|A(X)\big|
 \big|f\big(t + B(X)t\big) - f(t)\big|
 \,\big|\, \Dset \big] \\
&\leq \bigg(\bigg\{2\int_{-\infty}^{\infty}
 \big(1 + v^2\big)
 \Bigg( \frac{f'(v)}{f(v)}\bigg)^2\,F(dv)\bigg\}^{1/2}
 + 2\sup_{t \in \R} \big|tf(t)\big| \Bigg) \\
&\quad\times
 \Bigg( \sup_{t \in \R} \frac{|t|}{1 + t^2\min\{1,\,(1 + L)^2\}}
 \Bigg)
 E\big[\big|A(X)B(X)\big|\,\big|\,\Dset\big] \\
&\leq C_3 E\big[\big|A(X)B(X)\big|\,\big|\,\Dset\big],
\end{align*}
for large enough $n$, choosing $C_3$ proportional to
\begin{equation*}
\Bigg(\bigg\{2\int_{-\infty}^{\infty}
 \big(1 + v^2\big)\bigg(\frac{f'(v)}{f(v)}\bigg)^2\,F(dv)
 \bigg\}^{1/2} + 2\sup_{t \in \R} \big|tf(t)\big| \Bigg)
 \Bigg(\sup_{t \in \R} \frac{|t|}{1 + t^2\min\{1,\,(1 + L)^2\}}
 \Bigg).
\end{equation*}
Combining the calculations above, we have
\begin{align*}
&E\big[\big|A(X)B(X)\big| \,\big|\, \Dset\big] \\
&\leq E^{1/2}\big[|A(X)|^2 \,\big|\, \Dset\big]
 E^{1/2}\big[|B(X)|^2 \,\big|\, \Dset\big] \\
&\leq \sqrt{2}\bigg[\inf_{x \in [0,\,1]^m} \sigma(x)\bigg]^3
 \Bigg( \int_{[0,\,1]^m} \big|\ahat_1(x)\big|^2g(x)\,dx
 + \bigg[\sup_{x \in [0,\,1]^m} 
 \big|\rhat(x) - r(x) - \ahat_1(x)\big|\bigg]^2 \Bigg)^{1/2} \\
&\quad\times \Bigg(\frac12
 \int_{[0,\,1]^m}\big|\ahat_2(x)\big|^2g(x)\,dx
 + 4\bigg[\sup_{x \in [0,\,1]^m} \big|r(x)\big|\bigg]^2 
 \int_{[0,\,1]^m} \big|\ahat_1(x)\big|^2g(x)\,dx \\
&\quad\phantom{\leq\times\bigg(}
 + 2\bigg[\sup_{x \in [0,\,1]^m}
 \big|\rhat_2(x) - r_2(x) - \ahat_2(x)\big|\bigg]^2 \\
&\quad\phantom{\leq\times\bigg(}
 + 8\bigg[ \sup_{x \in [0,\,1]^m}
 \big\{\rhat(x) - r(x)\big\}^2 \bigg]^2
 + 8\bigg[ \sup_{x \in [0,\,1]^m}
 \big\{\shat(x) - \sigma(x)\big\}^2 \bigg]^2 \\
&\quad\phantom{\leq\times\bigg(}
 + 16\bigg[\sup_{x \in [0,\,1]^m} \big|r(x)\big|\bigg]^2
 \bigg[ \sup_{x \in [0,\,1]^m}
 \big|\rhat(x) - r(x) - \ahat_1(x)\big|\bigg]^2 \Bigg)^{1/2}.
\end{align*}
Again, it follows from the results above and \eqref{ahat12negli} for
this bound to be $\opn$, which implies \eqref{ABrem} is
$\opn$. Combining the above results for \eqref{Arem}, \eqref{Brem} and
\eqref{ABrem} shows that $\sup_{t \in \R} |R_2(t)| = \opn$.

It follows from the discussion in Remark 1 for analogous conclusions
of Proposition 1 to hold for the estimators $\rhat$ and $\shat$, i.e.\
\begin{equation*}
\bigg|\int_{[0,\,1]^m}\,\frac{\rhat(x) - r(x)}{\sigma(x)}g(x)\,dx
 - \avj e_j\bigg| = \opn
\end{equation*}
and
\begin{equation*}
\bigg|\int_{[0,\,1]^m}\,\frac{\shat(x) - \sigma(x)}{\sigma(x)}g(x)\,dx
 - \avj \frac{e_j^2 - 1}{2}\bigg| = \opn.
\end{equation*}
Since we have $\sup_{t \in \R}f(t) < \infty$ and $\sup_{t \in
  \R}|tf(t)| < \infty$, we find $\sup_{t \in \R}|R_3(t)| = \opn$. This
concludes the proof of Theorem 1.
\end{proof}


\section*{Acknowledgements}
Justin Chown gratefully acknowledges financial support from the
contract `Projet d'Actions de Recherche Concert\'ees' (ARC) 11/16-039
of the `Communaut\'e fran\c{c}aise de Belgique', granted by the
`Acad\'emie universitaire Louvain', and the Collaborative Research
Center ``Statistical modeling of nonlinear dynamic processes'' (SFB
823, Teilprojekt  C4) of the German Research Foundation (DFG). I would
like to thank the two referees for their careful reading of the
manuscript and the very helpful comments that greatly improved the
quality of this work.


\end{document}